%% file: main.tex
\PassOptionsToPackage{safe}{pbalance}
\documentclass[sigconf,nonacm,pbalance=true]{acmart}
\usepackage{amsmath}
\usepackage{booktabs}
\usepackage{multirow}
\usepackage{graphicx}
\usepackage{placeins}
\usepackage{needspace}
\usepackage{tikz}
\usetikzlibrary{arrows.meta}
\DeclareRobustCommand{\rev}[1]{#1}
\newenvironment{revision}{}{}
\AtBeginDocument{\hypersetup{hidelinks}}
\usepackage{enumitem}
\usepackage[ruled,vlined,linesnumbered]{algorithm2e}
\usepackage[nameinlink,noabbrev]{cleveref}
\Crefname{algocf}{Algorithm}{Algorithms}
\Crefname{AlgoLine}{Line}{Lines}
\crefname{algocf}{algorithm}{algorithms}
\crefname{AlgoLine}{line}{lines}
\Crefname{line}{Line}{Lines}
\AtEndPreamble{%
  \theoremstyle{acmplain}%
  \newtheorem{problem}[theorem]{Problem}}
\Crefname{problem}{Problem}{Problems}
\crefname{problem}{problem}{problems}
\Crefname{lemma}{Lemma}{Lemmas}
\Crefname{corollary}{Corollary}{Corollaries}
\makeatletter
\renewcommand{\nllabel}[1]{%
  \begingroup
  \def\@currentcounter{AlgoLine}%
  \protected@edef\cref@currentlabel{%
    [line][\arabic{AlgoLine}][\arabic{algocf}]\theAlgoLine}%
  \label[line]{#1}%
  \endgroup}
\makeatother
\SetKwInput{KwIn}{Input}
\SetKwInput{KwOut}{Output}
\setcopyright{none}
\acmConference[SIGMOD 2027]{SIGMOD Research Submission}{June 13--19, 2027}{Huntington Beach, CA, USA}
\acmYear{2027}
\acmDOI{}
\acmISBN{}
\newcommand{\method}{\texttt{CAP}}
\newsavebox{\maintablebox}
\input{sections/result_numbers}
\title{Scalable Approximate Algorithm for Dynamic Densest Subhypergraphs with Solution-Guided Maintenance}
\author{Jingbang Chen}
\authornote{Corresponding author.}
\affiliation{%
  \institution{CUHK-Shenzhen \& SLAI}
  \city{Shenzhen}
  \country{China}}
\author{Chenhao Ma}
\affiliation{%
  \institution{CUHK-Shenzhen}
  \city{Shenzhen}
  \country{China}}
\author{Yingli Zhou}
\affiliation{%
  \institution{CUHK-Shenzhen}
  \city{Shenzhen}
  \country{China}}
\renewcommand{\shortauthors}{Chen et al.}
\begin{document}
\begin{abstract}
\begin{revision}
Hypergraphs model interactions involving groups of entities, and finding
highly connected groups is a fundamental task in analyzing these data.
When interactions arrive and expire, maintaining a dense group can require
frequent and expensive changes to the underlying representation.
We propose \texttt{CAP}, a scalable algorithm that explicitly maintains a
$(1+\epsilon)$-approximate densest subhypergraph under hyperedge insertions
and deletions. The algorithm keeps a solution candidate together with an
endpoint allocation that bounds the optimum density.
The candidate's density guides maintenance: an update triggers repair only
when this bound is too large to establish the required approximation.
Local searches redistribute load or find a denser candidate, allowing
\texttt{CAP} to retain useful solutions across many updates.
We give an analysis that relates maintenance work to the size of the
regions searched and the density lost by the maintained solution.
It identifies conditions under which local repair remains inexpensive
and provides a theoretical explanation for the observed efficiency.
Experiments on real-world hypergraphs show speedups of up to nearly four
orders of magnitude over the evaluated dynamic baselines, while supporting
frequent queries and high accuracy. \texttt{CAP} also remains competitive
with specialized dynamic graph methods on ordinary graphs, with speedups
of up to approximately $64\times$ on the tested workloads.
\end{revision}
\end{abstract}
\ccsdesc[500]{Information systems~Data mining}
\ccsdesc[500]{Theory of computation~Dynamic graph algorithms}
\keywords{dynamic hypergraphs, densest subgraph, incremental maintenance, approximation certificates}
\maketitle
\input{sections/intro_related}
\input{sections/preliminaries}
\input{sections/cap_algorithm}

\input{sections/correctness}

\input{sections/work}

\input{sections/evaluation}

\FloatBarrier
\section{Conclusion}
\label{sec:conclusion}
Witness-conditioned repair gives a unified way to maintain a certified dense
set on graphs and hypergraphs. Its useful regime is described by certificate
slack, actual deletion-induced scale loss, and positive residual support.
The native-parameter experiments show substantial \rev{running-time} savings
against the released hypergraph implementation on the evaluated sources.
Static recomputation remains competitive at fine precision on smaller streams.
The ablation study shows that the combination of \rev{the larger outer capacity}
and reserve destinations reduces repair work on the evaluated hypergraphs.
\begin{acks}
The authors thank Sergey Kudria for early
discussions on the densest-subgraph problem in ordinary graphs.
\end{acks}

\section*{AI Acknowledgment}
The authors acknowledge the use of OpenAI GPT Astra for assistance with
literature search, algorithm and proof development, implementation,
experimental orchestration, analysis, and manuscript drafting.
The authors take full responsibility for the correctness of the mathematical
claims and proofs, the experimental results, and all other content of this paper.

\bibliographystyle{ACM-Reference-Format}
\begin{revision}
\bibliography{references,hypergraph_deep_refs}
\end{revision}
\input{appendix}

\end{document}

%% file: sections/result_numbers.tex
\newcommand{\MainSpeedup}{15.5--1869.2}
\newcommand{\NativeComparisonCells}{15}

%% file: sections/intro_related.tex
\section{Introduction}
\label{sec:cap-introduction}
\begin{revision}
A hypergraph represents relationships among groups of entities: its vertices
are entities, and each hyperedge connects all participants in one
interaction. Unlike an ordinary graph edge, a hyperedge can contain more
than two vertices. This representation occurs naturally in scientific
collaboration, where a paper joins its coauthors, and in online activity,
where an event associates a group of users or tags~\cite{Benson-2018-simplicial}.
Hypergraphs also model protein complexes~\cite{klamt2009}
and support learning from relationships among multiple
objects~\cite{feng2019hgnn}.
Keeping each group interaction intact preserves information that a
collection of pairwise links can lose, making hypergraphs a useful model
across social, biological, and information systems~\cite{battiston2020}.

A central task in such systems is to find a densely connected group.
In an ordinary graph, the densest-subgraph problem seeks a set of vertices
with the largest ratio of edges inside the set to the number of vertices
in it. This simple objective balances the number of interactions captured
against the size of the group. Densest-subgraph methods and their variants
are used to identify
communities~\cite{chen2012dense}, detect coordinated or anomalous
activity~\cite{hooi2016fraudar}, and identify regulatory motifs and
gene-annotation patterns~\cite{fratkin2006motifcut,saha2010restrictions,lanciano2024survey}.
The corresponding hypergraph problem counts hyperedges whose endpoints
all belong to the selected set. It finds groups participating in many
complete interactions, rather than merely many pairwise connections.
Algorithms for this objective include peeling, flow, and localized
methods~\cite{hd_cqtorres2022,hd_huang2024hyper,hd_bengali2026}.
In temporal data, however, interactions are continually added and removed,
so the dense group must be maintained as the hypergraph changes.
This dynamic setting has motivated algorithms for evolving
hypergraphs~\cite{hd_hu2017,hd_bera2022}.

Existing dynamic approaches maintain degree decompositions, orientations,
or endpoint allocations whose local conditions imply an approximation guarantee.
Hu et al.~\cite{hd_hu2017} give guarantees depending on hyperedge size;
Bera et al.~\cite{hd_bera2022} obtain a near-one approximation and provide
a practical implementation. Chekuri et al.~\cite{chekuri2024}
develop faster dynamic orientation techniques.
Their maintenance rules restore local conditions after updates.
Our key observation is that an update can violate such a condition while
the current solution still satisfies the requested approximation.
Retaining useful dense sets already helps ordinary-graph methods avoid
some recomputations~\cite{epasto2015,xu2024}.
For dynamic hypergraphs, we ask how the maintained solution can guide
allocation repair while preserving the required approximation after
every update.

We address this opportunity with \texttt{CAP}, a fully dynamic algorithm
for maintaining an explicit approximate densest subhypergraph.
It stores a solution candidate and distributes each hyperedge's contribution
among its endpoints. This allocation provides an upper bound on the
optimum density, while the candidate provides a lower bound---an
established relationship~\cite{hd_bera2022,hd_billig2025}.
After an insertion or deletion, \texttt{CAP} updates the affected
allocation and the candidate's density. If the two bounds remain close
enough, it keeps the candidate and finishes the update.
Otherwise, a local search either moves load to vertices with spare
capacity or identifies a denser set to use as the new candidate.
The solution is maintained explicitly, so a query can directly return
its vertices. Maintenance is solution-guided because the candidate's
current density guides both when repair is needed and when it can stop.

We also develop an analysis tailored to this maintenance policy.
A bound that charges every update for the whole hypergraph misses the
savings from retaining a useful solution and searching only a local region.
Our analysis separates the cost of an individual search from the amount
of repair caused by changes in solution density. For example, deleting
an edge outside the maintained solution does not lower its density and
cannot trigger an overload; when the solution loses only a small fraction
of its internal edges, its allowable load changes only slightly.
Together with bounds on the high-load region reached by a search, this
explains when the maintenance cost can be small even in a large hypergraph.
These conditions provide theoretical support for the strong performance
observed in our experiments.

Our contributions are summarized as follows.
\begin{itemize}[leftmargin=*]
  \item We propose \texttt{CAP}, a solution-guided algorithm that maintains
  an explicit $(1+\epsilon)$-approximate densest subhypergraph after every
  insertion or deletion. It supports hyperedges of different sizes and
  returns the maintained solution directly at queries.
  \item We analyze the work actually performed by local maintenance.
  The resulting bounds identify small high-load regions and limited
  decreases in solution density as conditions for inexpensive repair,
  explaining how the algorithm can benefit from the structure of a stream.
  \item We evaluate \texttt{CAP} against dynamic hypergraph baselines and
  exact snapshot solvers. It achieves speedups of up to approximately
  $5.9\times10^3$ over the dynamic baselines on completed comparisons,
  with studies of query frequency, data size, accuracy, and its maintenance
  policy. It also remains competitive with specialized methods on ordinary
  graphs, reaching speedups of approximately $64\times$.
\end{itemize}

\paragraph{Outline.}
\Cref{sec:cap-related} reviews related work, and \Cref{sec:cap-model}
defines the problem and notation. \Cref{sec:cap-method} presents
\texttt{CAP} and establishes its correctness.
\Cref{sec:hyper-support-candidate} analyzes solution-guided maintenance,
and \Cref{sec:evaluation} evaluates its efficiency and solution quality.
\Cref{sec:conclusion} concludes the paper.
\end{revision}

\section{Related Work}
\label{sec:cap-related}
\begin{revision}
\paragraph{Densest-subgraph discovery.}
Densest-subgraph discovery has a substantial literature spanning theory,
database algorithms, and applications, as covered by tutorials,
surveys, and comparative studies
\cite{gionis2015tutorial,fang2022tutorial,luo2023survey,lanciano2024survey,zhou2024analysis}.
Beyond Charikar's classical greedy $2$-approximation~\cite{charikar2000},
practical methods use iterative peeling~\cite{boob2020flowless},
core-based pruning~\cite{fang2019,xu2024}, convex optimization
\cite{danisch2017convex,harb2022scalable}, spectral methods
\cite{feng2024unified}, and parallel computation~\cite{luo2023,sukprasert2024parallel}.
Other work studies local algorithms~\cite{andersen2010local},
density-friendly decompositions~\cite{tatti2019density},
locally densest subgraphs~\cite{qin2015local,yang2026local,trung2023local},
and indexing all densest solutions~\cite{chang2020deconstruct}.
Directed densest-subgraph discovery has also received sustained attention,
including core-based and convex approaches by Ma et al.
\cite{ma2021directed,ma2024directed} and recent algorithms by
Zhou et al. and Liang et al.~\cite{zhou2025directed,liang2026directed}.
Variants impose size constraints~\cite{andersen2009size,khuller2009dense},
anchor the solution around specified vertices~\cite{dai2022anchored,zhou2026anchored},
or measure density through cliques~\cite{tsourakakis2015clique,zhou2024clique,zhou2025clique}.
Clique-based objectives also support locally dense solutions
\cite{xu2024localclique,zhou2026localclique}.

\paragraph{Densest subhypergraphs.}
Chekuri et al.~\cite{hd_cqtorres2022} study peeling and flow algorithms
through supermodular density, which includes fully contained-edge
hypergraph density. Huang et al.~\cite{hd_huang2024hyper} develop
formulations based on negative supermodularity and strongly localized methods,
along with exact global solvers. Bengali et al.~\cite{hd_bengali2026}
allow arbitrary monotonic partial-edge rewards and give an exact flow
implementation for the convex regime. Our objective is the special
case in which only fully contained hyperedges contribute.
These static solvers can answer a dynamic query by recomputing on the
current snapshot, but do not reuse the maintained allocation across
updates.

\paragraph{Dynamic densest subhypergraphs.}
Hu et al.~\cite{hd_hu2017} maintain dense subsets in evolving hypergraphs,
with approximation factors $r(1+\epsilon)$ for incremental updates
and $r^2(1+\epsilon)$ for fully dynamic updates.
Bera et al.~\cite{hd_bera2022} obtain a near-one factor independent of
hyperedge rank, including weighted hypergraphs, and provide the
\texttt{Udshp} implementation. Their algorithms use orientations and
local endpoint-load conditions.
A preliminary version of Chekuri et al.~\cite{chekuri2024} discusses
a hypergraph extension that reduces update and space costs using compact representations of
replicated edges and local labels; it also maintains an implicit
vertex-list representation of an approximate dense set.
We compare with \texttt{Udshp} and a reconstruction of this hypergraph
extension. Leng et al.~\cite{hd_leng2026decomposition} study hypergraph
density decomposition and give algorithms for insertions and deletions.
Their maintained object is a hierarchy of integer density layers.
\texttt{CAP} instead maintains an explicit approximate global solution
and uses its density to control changes to an endpoint allocation.

\paragraph{Dynamic densest subgraphs and orientations.}
For ordinary graphs, Bhattacharya et al.~\cite{bhattacharya2015} maintain
approximate density through hierarchical decompositions, while
Epasto et al.~\cite{epasto2015} retain dense sets to reduce
recomputation in evolving graphs. Sawlani and Wang~\cite{sawlani2020}
establish a near-one fully dynamic approximation through orientations.
The adaptive-orientation results of Chekuri et al.~\cite{chekuri2024}
and their engineering study~\cite{grossmann2025esa} further develop
local balance approaches for graphs.
Greedy dynamic orientations and their implementations have also been
studied~\cite{berglin2020,borowitz2023}.
Xu et al.~\cite{xu2024} maintain a density-threshold core containing
an optimum, using a retained dense set and local tests to avoid some
recomputations. A solver can refine that candidate into an exact or
approximate answer. Thus, reusing a useful solution is established
prior work; our policy uses its density to govern a persistent
allocation and maintain a near-one guarantee after every update.
Dynamic variants maintain solutions on directed graphs~\cite{ma2021directed},
dense regions under edge-weight changes~\cite{angel2014},
top-$k$ solutions~\cite{nasir2017}, dense subgraphs
in graph collections~\cite{valari2012}, temporal subgraphs
\cite{liu2022temporal}, anchored solutions~\cite{zhang2024anchored},
and density decompositions~\cite{zhang2025decomp}.
Exact dynamic orientations use directed searches and can skip changes
when the maximum degree remains feasible~\cite{grossmann2025exact}.
Pseudoarboricity maintenance~\cite{zhang2024pseudo} provides another
specialized graph baseline. Related static work computes orientations
minimizing the maximum indegree~\cite{venkateswaran2004}.
\end{revision}

%% file: sections/preliminaries.tex
\section{Preliminaries}
\label{sec:cap-model}
In this section, we introduce the notation and specify the dynamic
densest-subhypergraph problem. We then give its linear programming
formulation and discuss the allocation bounds and residual paths used in our algorithm.

\subsection{Hypergraphs and Density}
\label{sec:hypergraph-notation}
A hypergraph $G=(V,E)$ consists of a vertex set $V$ and a collection of
hyperedges $E$. Each hyperedge $e$ is a nonempty subset of $V$.
Let $n=|V|$ and $m=|E|$. We assume $n\ge1$ and fix an upper bound
$r\ge2$ on hyperedge size. The incidence volume is $I=\sum_{e\in E}|e|$, and $\deg(v)$ is
the number of hyperedges containing $v$. An ordinary graph is the special
case in which every edge has two endpoints ($r=2$). Our model also permits
singleton hyperedges and parallel hyperedges. Parallel hyperedges have
distinct identifiers and contribute separately to all edge counts.

For $S\subseteq V$, let $E(S)=\{e\in E:e\subseteq S\}$ be its induced
hyperedges. The density of a nonempty set $S$ and the optimal density are
\[
 \rho(S)=\frac{|E(S)|}{|S|},\qquad
 \rho^*=\max_{\emptyset\ne S\subseteq V}\rho(S).
\]
Thus, a hyperedge contributes one to the numerator exactly when all its
endpoints belong to $S$. For convenience, we set $\rho(\emptyset)=0$.

\subsection{Problem Specification}
\label{sec:problem-specification}
We consider a fixed vertex universe and a sequence of hyperedge insertions
and deletions. An insertion introduces a new edge identifier together with
its endpoints; a deletion removes an active identifier. The sequence may
contain arbitrary legal insertions and deletions.
All graph quantities below refer to the current snapshot; we omit the
time index when the snapshot is clear.

\begin{problem}[Dynamic densest subhypergraph]
\label[problem]{prob:dynamic-density}
Given $0<\epsilon\le1$, maintain a vertex set $C$ after each update such
that $\rho(C)\ge\rho^*/(1+\epsilon)$. A query reports the vertices of $C$
and their density. When $E=\emptyset$, the answer is $C=\emptyset$.
\end{problem}

We call the maintained set $C$ a \emph{witness}: its density is an explicit
lower bound on the optimum. We write $k=|C|$, $M=|E(C)|$, and
$L=\rho(C)$, so $L=M/k$ when $C\ne\emptyset$ and $L=0$ when
$C=\emptyset$. The algorithm also maintains an upper bound $U$ on
$\rho^*$. The condition $U\le(1+\epsilon)L$ then certifies the required
approximation. The witness is maintained after every update, even when
queries occur less frequently. Reporting its vertices costs $\Theta(1+k)$,
including constant time for an empty answer.

\subsection{\rev{Linear Programming Formulation}}
\label{sec:density-lp}
At a fixed snapshot, the densest-subhypergraph problem has the following
linear programming formulation~\cite{hd_bera2022}. A variable $y_v$ gives
the weight of vertex $v$, and $z_e$ gives the weight of hyperedge $e$:
\begin{equation}
\begin{aligned}
 \text{(P)}\qquad \max\quad &\sum_{e\in E}z_e\\
 \text{s.t.}\quad &z_e\le y_v &&(e\in E,\ v\in e),\\
 &\sum_{v\in V}y_v\le1,\\
 &y_v,z_e\ge0.
\end{aligned}
\label{eq:density-primal}
\end{equation}
The first constraint allows an edge to contribute only to the extent
supported by every endpoint. The second normalizes the total vertex
weight, turning the density ratio into a linear objective.

\begin{lemma}[Exact density formulation~\cite{hd_bera2022}]
\label[lemma]{lem:density-lp}
The optimum of \Cref{eq:density-primal} equals $\rho^*$.
\end{lemma}
Thus the linear program captures the original density objective exactly.
This equality lets us certify a combinatorially maintained set by
comparing its density with a feasible dual allocation.
We interpret the dual of \Cref{eq:density-primal} as an orientation.
Associate $f_{e,v}\ge0$ with each constraint $z_e\le y_v$ and $D\ge0$
with the normalization constraint. The dual is
\begin{equation}
\begin{aligned}
 \text{(D)}\qquad \min\quad &D\\
 \text{s.t.}\quad &\sum_{v\in e}f_{e,v}\ge1 &&(e\in E),\\
 &\sum_{e\ni v}f_{e,v}\le D &&(v\in V),\\
 &f_{e,v}\ge0.
\end{aligned}
\label{eq:density-dual}
\end{equation}
For each edge, excess mass above one can be removed without increasing
any vertex load. Thus an optimal dual solution can satisfy
$\sum_{v\in e}f_{e,v}=1$. Such an endpoint allocation is a
\emph{fractional orientation}: an integral orientation assigns the
whole unit to one incident vertex, while a fractional orientation can
split it among several vertices. By strong duality, the minimum possible
maximum load over fractional orientations equals $\rho^*$.
Thus, any induced vertex set gives a lower bound on the optimum, while
any fractional orientation gives an upper bound through its maximum load.

\subsection{Allocation Bounds and Residual Paths}
\label{sec:allocation-basics}
We next describe how to store the orientation, use it to certify
an approximate solution, and adjust it through residual paths.
We represent a fractional orientation using integer endpoint shares.
Following the integer representation of fractional
orientations~\rev{\cite{chekuri2024}}, fix a positive integer $b$ and
store shares $x_{e,v}\in\mathbb Z_{\ge0}$ satisfying
\[
 \sum_{v\in e}x_{e,v}=b,\qquad q_v=\sum_{e\ni v}x_{e,v}.
\]
The fraction assigned to $v$ by $e$ is $x_{e,v}/b$, and $q_v$ is its
scaled load. Setting $f_{e,v}=x_{e,v}/b$ gives a feasible solution to
\Cref{eq:density-dual} with $D=\max_vq_v/b$.
These quantized allocations form a subset of the dual feasible solutions;
strong duality refers to optimization over all fractional orientations.
Shares are counters; they do not require $b$ copies of an edge.

\begin{lemma}[Allocation upper bound~\cite{hd_bera2022,hd_billig2025}]
\label[lemma]{lem:allocation-upper}
Every feasible allocation satisfies $\rho^*\le U$, where
$U=\max_{v\in V}q_v/b$.
\end{lemma}
The allocation need not minimize its maximum load: any feasible one
already gives a valid upper bound.
To certify the approximation, we compare this upper bound with
the density of the maintained set. The witness $C$ supplies a feasible primal solution of value $L$, and
the maintained orientation supplies a feasible dual solution of value $U$.
The approximation certificate is therefore the primal--dual bound
\[
 L\le\operatorname{OPT}(\mathrm P)=\rho^*
   =\operatorname{OPT}(\mathrm D)\le U\le(1+\epsilon)L.
\]
In \Cref{sec:cap-method}, we maintain this certificate through
combinatorial changes to the witness and endpoint shares.

For example, suppose $L=3$ and $U=3.2$, as illustrated in
\Cref{fig:prelim-certificate}(a). Since $3\le\rho^*\le3.2$ and
$3.2\le(1+0.1)\cdot3$, the maintained set is already a
$(1+0.1)$-approximation. The certificate does not require knowing
$\rho^*$ exactly or further reducing $U$.

\paragraph{Residual paths.}
A \emph{positive incidence} is a pair $(e,u)$ with $x_{e,u}>0$.
It permits a residual step $u\xrightarrow{e}v$ to any $v\in e$:
some of $e$'s share can move from $u$ to $v$.
Consider a residual path of length $\ell\ge1$ with distinct vertices
and distinct hyperedges:
\[
v_0\xrightarrow{e_1}v_1\xrightarrow{e_2}\cdots
\xrightarrow{e_\ell}v_\ell.
\]
Moving a positive integer amount $a\le\min_i x_{e_i,v_{i-1}}$
at every step preserves
each edge's total allocation. The load of $v_0$ decreases by $a$, the
load of $v_\ell$ increases by $a$, and the intermediate loads are unchanged.
\Cref{fig:prelim-certificate}(b) illustrates a two-edge path:
transferring one share along each edge changes the three loads from
$(16,15,14)$ to $(15,15,15)$.
This is the basic operation used to repair an allocation. Residual
closure under positive incidences is also used in hypergraph extraction
algorithms~\cite{hd_billig2025}.
\input{sections/preliminaries_figure}

%% file: sections/preliminaries_figure.tex
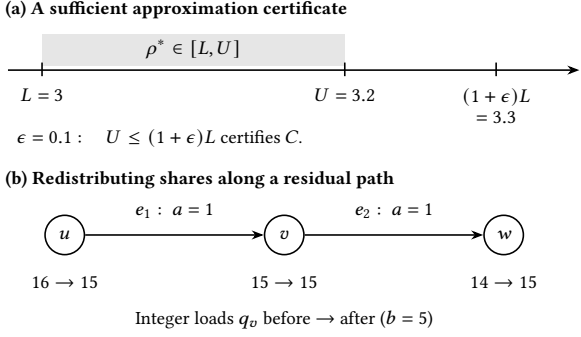
\begin{figure}[t]
\centering
\begin{tikzpicture}[x=1cm,y=1cm,draw=black,text=black,
  every node/.style={font=\footnotesize},line width=0.6pt,
  >={Stealth[length=1.6mm,width=1.2mm]}]
\path[use as bounding box] (0,-3.5) rectangle (8.0,1.05);
\node[anchor=west,font=\footnotesize\bfseries] at (0,0.9)
  {(a) A sufficient approximation certificate};
\fill[black!10] (0.6,0.16) rectangle (4.6,0.57);
\node at (2.6,0.38) {$\rho^*\in[L,U]$};
\draw[->] (0.15,0.1) -- (7.7,0.1);
\foreach \x in {0.6,4.6,6.6} {\draw (\x,0.02) -- (\x,0.2);}
\node[anchor=north] at (0.6,-0.02) {$L=3$};
\node[anchor=north] at (4.6,-0.02) {$U=3.2$};
\node[anchor=north,align=center] at (6.6,-0.02)
  {$(1+\epsilon)L$\\$=3.3$};
\node[anchor=west] at (0.15,-0.8)
  {$\epsilon=0.1:\quad U\le(1+\epsilon)L$ certifies $C$.};
\node[anchor=west,font=\footnotesize\bfseries] at (0,-1.35)
  {(b) Redistributing shares along a residual path};
\foreach \name/\x in {u/0.9,v/3.8,w/6.7} {
  \node[circle,draw=black,minimum size=5.2mm,inner sep=0pt]
    (\name) at (\x,-2.08) {$\name$};
}
\draw[->] (u) -- node[above=3pt] {$e_1:\ a=1$} (v);
\draw[->] (v) -- node[above=3pt] {$e_2:\ a=1$} (w);
\node at (0.9,-2.7) {$16\to15$};
\node at (3.8,-2.7) {$15\to15$};
\node at (6.7,-2.7) {$14\to15$};
\node at (3.8,-3.2) {Integer loads $q_v$ before $\to$ after ($b=5$)};
\end{tikzpicture}
\caption{Allocation tools. (a) The interval $[L,U]$ contains the
optimal density and certifies the maintained set at $\epsilon=0.1$.
(b) A valid residual path on distinct hyperedges transfers one share
at each step; the middle load is unchanged. Arrows show share transfers,
not directions of the input hyperedges.}
\label{fig:prelim-certificate}
\end{figure}

%% file: sections/cap_algorithm.tex
\section{Witness-Guided Dynamic Maintenance}
\label{sec:cap-method}
We formally propose our main algorithm called \method{}
(\emph{\underline{C}apacitated \underline{A}ugmenting \underline{P}aths}) for
\Cref{prob:dynamic-density}. We now give an overview of the algorithm.

For a requested accuracy $0<\epsilon\le1$, we fix the integer precision
$b\ge\lceil2r/\epsilon\rceil$ and two multipliers
$\alpha=1+\epsilon$ and $\gamma=1+\epsilon/2$.
The lower bound on $b$ ensures $bL\ge2/\epsilon$ after witness
normalization (\Cref{sec:cap-initialization}).
For a nonempty witness of density $L$, define an \emph{outer capacity}
$c=\lfloor\alpha bL\rfloor$ and an \emph{inner threshold}
$t=\lfloor\gamma bL\rfloor$.
The outer capacity $c$ specifies the stopping condition: repair is
needed exactly when some $q_v>c$. The inner threshold $t$ specifies
where to move load: the search seeks a destination with $q_v<t$,
although the transfer may increase its load up to $c$. We call this
choice \emph{reserve destinations}.

The main idea is to use the density of the current witness to decide
when the allocation needs repair.
For a query, we directly output the maintained witness $C$.
By \Cref{lem:allocation-upper}, a
maximum load of at most $(1+\epsilon)bL$ already certifies $C$.
If an update violates this condition, there are two ways to restore it:
reduce the maximum load by moving shares, or increase the density of the
witness. A residual search (\Cref{alg:cap-search}) identifies which of
these actions to take. After every edge update processed by
\Cref{alg:cap-update}, the maintained set has density at least
$\rho^*/(1+\epsilon)$ (\Cref{thm:certificate}).

Specifically, \Cref{alg:cap-update} has three steps. First, it updates the shares of
the inserted or deleted edge and the induced-edge count of $C$
(\Cref{sec:cap-update}). Second, it normalizes a witness whose density
is too small and computes the resulting capacity
(\Cref{sec:cap-refresh}). Third, while a load exceeds this capacity,
it searches from an overloaded vertex (\Cref{sec:cap-search}). A
successful search returns a path for transferring load. If no destination
is reachable, the reached set becomes a denser witness and increases the
inner threshold; the outer capacity cannot decrease. An occasional peeling pass can also improve the witness
(\Cref{sec:cap-refresh}). The procedure returns once all loads are at most
$c$; \Cref{sec:correctness} proves that it always reaches this point.
We analyze the maintenance cost in \Cref{sec:hyper-support-candidate}.

\subsection{Initialization}
\label{sec:cap-initialization}
To implement \method{}, we store the shares of each active edge, the loads $q_v$, and
an indexed maximum heap over these loads. We also store $C$, membership
marks, and the exact induced-edge count $M$. Each vertex has a list of its
positive incidences, with inverse positions for updating the list when a
share becomes zero or positive. Full incidence lists support edge updates
and witness refresh. A counter $D$ records positive-incidence inspections
since the last refresh. Thresholds are evaluated by exact integer
cross-products; the analysis assumes that all intermediate arithmetic
values are representable.

We call a state \emph{stable} when the allocation is feasible, the
witness count is exact, and all loads are at most $c$; the empty state
is stable with zero loads and zero bounds. At the end of each update,
the stable allocation and witness satisfy
\begin{equation}
 L\le\rho^*\le U\le\alpha L.
 \label{eq:cap-certificate}
\end{equation}

Starting from an empty graph, we set all loads to zero, $C=\emptyset$,
$M=L=U=0$, and $D=0$. An initial nonempty graph is built by inserting its
edges through \Cref{alg:cap-update}. Before repairing a nonempty
graph, we ensure $L\ge1/r$. If necessary, we use the endpoints of a
surviving edge as the witness and count all edges it induces. This
normalization guarantees $bL\ge2/\epsilon$.

\begin{algorithm}[t]
\small
\caption{\texttt{CAP}: process an edge update}
\label{alg:cap-update}
\KwIn{An insertion or deletion of edge $e$; maintained state}
\KwOut{A witness $C$ and bounds $L,U$}
\eIf{inserting $e$\nllabel{ln:update-branch}}{
  $M\gets M+\mathbf{1}[e\subseteq C]$\nllabel{ln:update-insert-count}\;
  Insert $e$ and water-fill its $b$ units among its endpoints\nllabel{ln:update-insert-mass}\;
}{
  $M\gets M-\mathbf{1}[e\subseteq C]$\nllabel{ln:update-delete-count}\;
  Subtract $x_{e,v}$ from every $q_v$, $v\in e$, and remove $e$\nllabel{ln:update-delete-mass}\;
}
\If{$E=\emptyset$\nllabel{ln:update-empty-test}}{Set $C\gets\emptyset$, $M,L,U\gets0$; \Return\nllabel{ln:update-empty}\;}
Normalize $C$ if $C=\emptyset$ or $M/|C|<1/r$\nllabel{ln:update-normalize}\;
Set $L\gets M/|C|$, $c\gets\lfloor\alpha bL\rfloor$; $f\gets\mathrm{false}$\nllabel{ln:update-capacity}\;
\While{$\max_vq_v>c$\nllabel{ln:update-repair-loop}}{
  \If{refresh is enabled, $D>4(n+rm)$, and $f=\mathrm{false}$\nllabel{ln:update-refresh-test}}{
    Refresh $C$ by peeling; $D\gets0$; $f\gets\mathrm{true}$\nllabel{ln:update-refresh}\;
    Recompute $L,c$\nllabel{ln:update-refresh-capacity}\; \If{$\max_vq_v\le c$}{\textbf{break}\nllabel{ln:update-refresh-done}\;}
  }
  $z\gets\arg\max_vq_v$; $t\gets\lfloor\gamma bL\rfloor$\nllabel{ln:update-root}\;
  Search from $z$ with threshold $t$ using \Cref{alg:cap-search}\nllabel{ln:update-search}\;
  Add the search's positive-incidence inspections to $D$\nllabel{ln:update-debt}\;
  \eIf{the search returns a path $\pi$ to $w$}{
    Transfer the amount in \Cref{eq:cap-transfer} along $\pi$\nllabel{ln:update-transfer}\;
  }{
    $C\gets R$; $M\gets\sum_{v\in R}q_v/b$\nllabel{ln:update-witness}\;
    Recompute $L\gets M/|C|$ and $c\gets\lfloor\alpha bL\rfloor$\nllabel{ln:update-witness-capacity}\;
  }
}
$U\gets\max_vq_v/b$; \Return $(C,L,U)$\nllabel{ln:update-return}\;
\end{algorithm}

\subsection{Processing Insertions and Deletions}
\label{sec:cap-update}
The implementation details of processing edge updates are shown in
\Cref{alg:cap-update}. The procedure consists of the following steps:
\begin{itemize}[itemsep=2pt,topsep=3pt]
\item[(1)] Applying the update
(\Crefrange{ln:update-branch}{ln:update-delete-mass}).
\item[(2)] Preparing the certificate
(\Crefrange{ln:update-empty-test}{ln:update-capacity}).
\item[(3)] Restoring the certificate
(\Crefrange{ln:update-repair-loop}{ln:update-witness-capacity}).
\item[(4)] Returning the solution (\Cref{ln:update-return}).
\end{itemize}

\subsubsection{Applying the Update}
For an insertion, we increase $M$ exactly when all endpoints of $e$ lie
in $C$ (\Cref{ln:update-insert-count}). We then insert the edge and
allocate its $b$ units by \emph{discrete water filling}
(\Cref{ln:update-insert-mass}). This rule gives priority to the
least-loaded endpoints of $e$. More precisely, let $q_v$ denote the
load before insertion and start with $x_{e,v}=0$ for every $v\in e$.
Repeatedly increase $x_{e,v}$ by one for an endpoint minimizing
$q_v+x_{e,v}$, breaking ties by endpoint order, until
$\sum_{v\in e}x_{e,v}=b$. We then add these shares to the endpoint
loads. A binary search for the filling level computes the same
allocation in batches.

For a deletion, we decrease $M$ under the same
containment condition (\Cref{ln:update-delete-count}) and remove the
edge's shares from the endpoint loads (\Cref{ln:update-delete-mass}).
Both branches update the affected incidence lists and heap keys.

\subsubsection{Preparing the Certificate}
If no edge remains, we return the empty witness and zero bounds
(\Cref{ln:update-empty}). Otherwise, we normalize the witness when
needed (\Cref{ln:update-normalize}): we replace an empty witness or one
of density below $1/r$ by the endpoints of a surviving edge and compute
their induced-edge count. \Cref{sec:cap-normalization} gives the details.
We then compute $L$ and $c$, and reset $f$ to false
(\Cref{ln:update-capacity}). This flag permits at most one peeling
refresh during the current update.

An insertion never lowers the density of the retained witness. Only
the inserted edge adds load, so it creates at most $b$ units of excess
above the new capacity. Deleting an edge $e\nsubseteq C$ leaves the
witness density and capacity unchanged while reducing loads; it
therefore requires no repair. Deleting an edge $e\subseteq C$ can
lower the capacity throughout the graph. We consequently test the global maximum load after normalization
and enter the repair loop only when this maximum exceeds $c$
(\Cref{ln:update-repair-loop}).

\subsubsection{Restoring the Certificate}
Each iteration first checks whether the accumulated inspection count
$D$ exceeds $4(n+rm)$ and no refresh has occurred during this update
(\Cref{ln:update-refresh-test}). If so, a peeling pass seeks a denser
witness, which can raise the capacity without moving load
(\Cref{ln:update-refresh}; details in \Cref{sec:cap-budgeted-refresh}).
We reset $D$, set $f$ to true, and recompute $L,c$. If the new capacity
already bounds every load, the loop ends
(\Crefrange{ln:update-refresh-capacity}{ln:update-refresh-done}).

Otherwise, we select a maximum-load vertex $z$, compute the inner
threshold $t$, and invoke the residual search
(\Cref{ln:update-root,ln:update-search}). This search either finds a
path from $z$ to a vertex $w$ with $q_w<t$, or returns the entire reached
set $R$ when no such destination exists. \Cref{sec:cap-search} describes
this breadth-first process in \Cref{alg:cap-search}. Loads remain fixed
during the search; its positive-incidence inspections are added to $D$
(\Cref{ln:update-debt}).

On success, we transfer load along the returned path
(\Cref{ln:update-transfer}).
Let $\pi$ be the parent path from $z$ to the destination $w$, and let
$x_i$ be the donor share on its $i$th step. We transfer
\begin{equation}
 a=\min\left\{q_z-c,\ c-q_w,\ \min_{i\in\pi}x_i\right\}
 \label{eq:cap-transfer}
\end{equation}
units along every step. On a step $u\xrightarrow{e}v$, this subtracts
$a$ from $x_{e,u}$ and adds $a$ to $x_{e,v}$. Positive-incidence lists
are adjusted when shares cross zero. Only the loads and heap keys of
$z$ and $w$ change. The destination may end above $t$; the transfer is
bounded by the outer capacity $c$.

For example, suppose $r=3$, $\epsilon=0.5$, $b=12$, and $L=1$, giving
$c=18$ and $t=15$. A root of load $20$ can transfer two units to a
destination of load $14$ if every donor share on the path is at least
two. The new endpoint loads are $18$ and $16$. The destination exceeds
$t$ after the transfer, but both endpoints satisfy the capacity $c$.

On failure, we replace $C$ by $R$ and obtain its exact edge count as
$M=\sum_{v\in R}q_v/b$ (\Cref{ln:update-witness}).
\Cref{lem:closed-witness} proves both this identity and $\rho(R)>L$.
We recompute the density and capacity
(\Cref{ln:update-witness-capacity}), then repeat the maximum-load test.
The new capacity cannot decrease, and the load sum avoids a separate
containment scan over the reached set.

\subsubsection{Returning the Solution}
Once all loads fit within $c$, we set $U=\max_vq_v/b$ and return the
maintained set $C$ together with $L$ and $U$ (\Cref{ln:update-return}).
The stopping condition and allocation bound give
$L\le\rho^*\le U\le(1+\epsilon)L$.
\Cref{sec:correctness} proves that every update reaches this condition.

\subsubsection{Complexity Analysis}
For an edge of size $s$, containment testing and share removal inspect
$O(s)$ incidences. Batch water filling takes
$O(s[\log(s+1)+\log(b+1)])$ time, and updating the heap keys costs
$O(s\log(n+1))$. Deletion also updates inverse incidence positions in
$O(s\log(r+1))$ time. These costs are incurred even when no residual
search is needed.

Restoring the certificate additionally pays for every search and path
transfer. \Cref{sec:cap-search-time} gives the search cost in terms of
the incidences examined. For a path of $\ell$ steps, the parent pointers
identify donor shares, and binary search locates receiving shares in
the sorted endpoint lists. Thus a transfer takes
$O(\ell\log(r+1)+\log(n+1))$ amortized time, including changes to the two endpoint
heap keys. Normalization and optional refresh are accounted for in
\Cref{sec:cap-normalization,sec:cap-budgeted-refresh};
\Cref{sec:total-maintenance-work} combines these costs.
Querying the bounds costs $O(1)$; reporting $C$ costs
$\Theta(1+|C|)$.

\subsection{Residual Search}
\label{sec:cap-search}
\label{sec:cap-reserve}
The residual search determines where an overloaded vertex can send
load. \Cref{alg:cap-search} performs a breadth-first search from $z$
along positive incidences, using the residual paths defined in
\Cref{sec:allocation-basics}. It stops at the first reached vertex whose
load is below the destination threshold $\theta$. If there is no such
vertex, it returns the reached set $R$. In \method{}, $\theta=t$;
the returned path permits a transfer, while a failed search identifies
a denser witness.

\begin{algorithm}[t]
\small
\caption{\texttt{Search}$(z,\theta)$}
\label{alg:cap-search}
\KwIn{Root $z$ with $q_z>c$; destination threshold $\theta\le c$}
\KwOut{A residual path, or the reached set $R$}
$R\gets\{z\}$; $Q\gets[z]$; mark no edge expanded\nllabel{ln:search-init}\;
\While{$Q\ne\emptyset$\nllabel{ln:search-loop}}{
  Pop the first vertex $u$ from $Q$\nllabel{ln:search-pop}\;
  \ForEach{positive incidence $(e,u)$\nllabel{ln:search-positive}}{
    \If{$e$ is already expanded}{\textbf{continue}\nllabel{ln:search-skip}\;}
    Mark $e$ expanded\nllabel{ln:search-mark}\;
    \ForEach{$v\in e$, in stored order\nllabel{ln:search-endpoints}}{
      \If{$v\notin R$\nllabel{ln:search-unseen}}{
        Add $v$ to $R$; set its parent to $(u,e)$\nllabel{ln:search-parent}\;
        \If{$q_v<\theta$\nllabel{ln:search-destination-test}}{\Return the parent path from $z$ to $v$\nllabel{ln:search-success}\;}
        Append $v$ to $Q$\nllabel{ln:search-enqueue}\;
      }
    }
  }
}
\Return $R$ as a failed search\nllabel{ln:search-failure}\;
\end{algorithm}

\subsubsection{Search Procedure}
\Cref{ln:search-init} initializes $R$ and the queue $Q$ with $z$ and
starts a fresh expansion stamp. The queue processes vertices in
breadth-first order (\Cref{ln:search-loop,ln:search-pop}). For each
popped vertex $u$, we inspect its positive incidences
(\Cref{ln:search-positive}), since only edges with $x_{e,u}>0$ can
transfer load away from $u$.

\Cref{ln:search-skip,ln:search-mark} skip an edge already expanded by
this search and mark a newly encountered edge before examining its
endpoints. Each edge is therefore expanded at most once. For every
unseen endpoint $v$, we add it to $R$ and record $(u,e)$ as its parent
(\Crefrange{ln:search-endpoints}{ln:search-parent}). These parent
pointers retain the residual steps by which vertices were reached.

If $q_v<\theta$, the parent pointers give a path from $z$ to $v$, which
the search returns immediately
(\Crefrange{ln:search-destination-test}{ln:search-success}). Otherwise,
$v$ is appended to $Q$ for further exploration
(\Cref{ln:search-enqueue}). If the queue becomes empty, every positive
incidence reachable from $z$ has been examined; we return $R$ as a
failed search (\Cref{ln:search-failure}).

\subsubsection{Effectiveness Analysis}
\label{sec:cap-search-effectiveness}
The following two lemmas show that successful searches reduce excess
load and failed searches improve the witness. The load gap provided by
reserve destinations also yields the transport bound in
\Cref{sec:reserve-transport}.

\begin{lemma}[Progress of a successful search]
\label[lemma]{lem:transfer-progress}
A successful transfer preserves a feasible allocation and decreases
the integer excess $X_c=\sum_v(q_v-c)_+$ by $a\ge1$.
\end{lemma}
Thus a successful search makes measurable progress: it removes $a$
units of overload without creating overload at the destination. This
bounds the number of successful searches by the total excess repaired.
\begin{proof}
The first expansion of an edge reaches all its unseen endpoints at
the same BFS layer, unless the search terminates. These endpoints are
siblings in the parent tree; a parent path thus uses an edge at most
once. Its minimum donor share bounds the transfer on each edge.
All intermediate load changes cancel. The root loses $a$ excess units,
and the destination receives at most its available capacity $c-q_w$.
\end{proof}

If no destination is found, the reached set provides a better witness.
This makes failed searches useful: they improve the witness and strictly
increase the inner threshold used by subsequent searches
(\Cref{eq:failure-threshold-growth}).

\begin{lemma}[Witness improvement]
\label[lemma]{lem:closed-witness}
A failed search with threshold $t$ returns a set $R$ of density
$\rho(R)>L$ satisfying
\begin{equation}
 b|E(R)|=\sum_{v\in R}q_v>t|R|,
\label{eq:cap-closed-count}
\end{equation}
\end{lemma}
The conclusion supplies the other kind of progress. When load cannot be
moved out of the reached set, that set itself improves the witness
density and strictly increases the inner threshold. The outer capacity
cannot decrease. The load sum also gives its induced-edge count,
so this improvement needs no separate containment scan.
\begin{proof}
Every edge with a positive owner in $R$ exposes all its endpoints
before the queue becomes empty. Hence crossing edges assign no mass
to $R$, whereas each induced edge assigns all $b$ units there.
Every reached vertex has load at least $t$, and the root exceeds
$c\ge t$, giving the strict inequality. Since $bL\ge2/\epsilon$,
$t=\lfloor bL+\epsilon bL/2\rfloor\ge bL$. Thus $\rho(R)>L$.
\end{proof}

\subsubsection{Time Analysis}
\label{sec:cap-search-time}
For one search, let $P$ count inspected positive incidences, including
those skipped because their edge was already expanded, and let $J$
count endpoints inspected during edge expansions. Vertex and edge
stamps avoid clearing all marks before a search. Every vertex other
than the root is discovered during an endpoint inspection, and is
queued at most once. Hence queue operations, marking, and parent
recording take $O(1+J)$ time; the two incidence loops take $O(P+J)$
time. A returned path has at most $J$ steps, so following its parent
pointers is covered by the same bound. Thus the search takes
$O(1+P+J)$ time. Summing the loads of a failed reached set also takes
$O(1+J)$ time. Updating the allocation and installing a new witness
are charged separately in \Cref{sec:total-maintenance-work}.

Each vertex is expanded at most once and each edge is stamped after
its first expansion, so $P\le I$ and $J\le I$ for the current incidence
volume $I$. This gives an $O(1+I)$ worst-case bound for a single search.
\Cref{sec:local-search-work} sharpens the bound using the high-load
region visited by the search; \Cref{sec:reserve-transport} bounds the
transport needed across successful searches.

We also bound the number of failed searches during one update.
By \Cref{lem:closed-witness}, a failure produces a new density $L'$
with $bL'>t$. The recomputed threshold therefore satisfies
$t'=\lfloor(1+\epsilon/2)bL'\rfloor\ge t+\lfloor\epsilon t/2\rfloor$.
Normalization gives $t\ge2/\epsilon$. Since
$\lfloor a\rfloor\ge a/2$ for $a\ge1$, each failure yields
\begin{equation}
 t'\ge t+\lfloor\epsilon t/2\rfloor\ge(1+\epsilon/4)t.
 \label{eq:failure-threshold-growth}
\end{equation}
If $t_0$ is the threshold before the first search in an update, its number
of failed searches is consequently at most
\[
 \left\lfloor
 \frac{\log\bigl(\lfloor(1+\epsilon/2)b\rho^*\rfloor/t_0\bigr)}
      {\log(1+\epsilon/4)}
 \right\rfloor.
\]
Here the graph is fixed during repair, so the threshold is at most
$\lfloor(1+\epsilon/2)b\rho^*\rfloor$ and cannot decrease between
searches. This count applies to a fixed update;
\Cref{sec:total-maintenance-work} accounts for threshold decreases
across an entire update sequence.

\subsection{Maintaining and Refreshing the Witness}
\label{sec:cap-refresh}
We discuss the details of two steps in \Cref{alg:cap-update}:
normalization (\Cref{ln:update-normalize}) and budgeted refresh
(\Cref{ln:update-refresh}). Normalization ensures the minimum witness
density needed by the thresholds, whereas refresh seeks a denser
witness to reduce subsequent repair work.

\subsubsection{Normalization}
\label{sec:cap-normalization}
The normalization step (\Cref{ln:update-normalize}) chooses a surviving
edge $e$ and sets $S=e$ when the witness is empty or $L<1/r$. The set may
induce several hyperedges, so we compute $|E(S)|$ exactly. We inspect
positive incidences at vertices of $S$, deduplicate encountered edges
with a stamp, and test each encountered edge for containment. A test
stops at the first endpoint outside $S$. Every induced edge has a
positive owner in $S$ and is counted once. Installing this witness
restores $L\ge1/|e|\ge1/r$.

\subsubsection{Budgeted Refresh}
\label{sec:cap-budgeted-refresh}
A denser witness can avoid subsequent repairs. We occasionally run
minimum-degree hypergraph peeling and select the densest remaining
suffix of its removal order. The suffix replaces $C$ only if its
density is strictly larger. During peeling, an edge is removed once,
when its first endpoint is removed, so suffix densities are exact.
The allocation is retained.

The refresh is triggered before a search when accumulated positive
scans exceed $4(n+rm)$. We install the improved witness, reset $D$, and
set $f$ to true (\Cref{ln:update-refresh}). Recomputing $L$ and $c$ may
already certify the current allocation, in which case the repair loop
ends (\Cref{ln:update-refresh-capacity}). At most one refresh occurs in an
update. A pass costs
$O((n+I)\log(2+n+I))$, which is charged to the accumulated scans in
\Cref{sec:total-maintenance-work}. Both normalization and refresh
never decrease the current witness density and therefore cannot create an
overload by lowering the capacity.

%% file: sections/correctness.tex
\subsection{Correctness}
\label{sec:correctness}
Updates and transfers preserve
feasible shares (\Cref{lem:transfer-progress}); successful
searches reduce excess load, while failed searches increase the witness
density (\Cref{lem:closed-witness}); and the stopping condition gives the
approximation certificate (\Cref{lem:allocation-upper}). The only
remaining issue is to show that repair cannot continue indefinitely.
The proof of \Cref{thm:certificate} below establishes this.

\begin{theorem}[All-update guarantee]
\label{thm:certificate}
Let $0<\epsilon\le1$ and $b\ge\lceil2r/\epsilon\rceil$.
For \method{}, every legal update terminates with an explicit set $C$ satisfying
$\rho(C)\le\rho^*\le U\le(1+\epsilon)\rho(C)$.
\end{theorem}
\begin{proof}
Insertion and removal preserve the allocation of surviving edges;
\Cref{lem:transfer-progress} establishes the same property for transfers.

During repair the graph is fixed. Define excess
$X_c=\sum_v(q_v-c)_+$. Every successful transfer decreases $X_c$ by its
positive integer amount. A witness improvement cannot increase $X_c$.
Recall that $t=\lfloor(1+\epsilon/2)bL\rfloor$.
Normalization gives $bL\ge2/\epsilon$, hence $t\ge2/\epsilon$.
On any failed search, the new witness density $L'$
satisfies $bL'>t$. Therefore
\[
 \lfloor(1+\epsilon/2)bL'\rfloor
 \ge t+\lfloor\epsilon t/2\rfloor\ge t+1.
\]
The witness strictly improves, and the inner threshold increases strictly.
It is bounded by $\lfloor(1+\epsilon/2)b\rho^*\rfloor$ on this fixed
graph. There are finitely many failures, at most the initial excess in
successful searches, and at most one optional refresh per update.
Repair thus terminates with $\max_vq_v\le c$, which yields the claimed
certificate by allocation duality and $c=\lfloor(1+\epsilon)bL\rfloor$.
\end{proof}
We summarize the time and space complexity in \rev{\Cref{sec:hyper-support-candidate}}.
Deriving the maintenance-time bound requires quantifying the search work
and load transport of witness-guided repair, so we devote
\Cref{sec:hyper-support-candidate} to this analysis.

%% file: sections/work.tex
\section{Analysis of Solution-Guided Maintenance}
\label{sec:hyper-support-candidate}
\Cref{thm:certificate} proves that \method{} maintains a
$(1+\epsilon)$-approximate solution after every update.
We now analyze the work required to maintain this guarantee.
A bound that charges every update
for the whole hypergraph hides the main benefit of \method{}: an update
can finish without a search, and a search can stop after inspecting only
a small part of the allocation. We analyze these savings through the
work performed by the maintenance procedure.
We first give a general update-time guarantee, then refine it to
identify the work saved by solution-guided maintenance.

Consider $N\ge1$ legal updates from an empty hypergraph, with
$b=\lceil2r/\epsilon\rceil$ and $0<\epsilon\le1$.
A nonempty initial graph is charged as a sequence of insertions.

\begin{theorem}[General amortized update time]
\label{thm:amortized-update}
\method{} has amortized
update time
\[
 O\!\left(r m_* I_*\,\frac{\log(n+r m_*)+\log(1/\epsilon)}{\epsilon}\right),
\]
where $m_*$ and $I_*$ are the maximum active hyperedge count and
incidence volume over the update sequence, respectively.
The bound excludes the one-time $O(n)$ initialization cost.
\end{theorem}

\Cref{thm:amortized-update} follows by bounding every search using
the full active incidence volume and every update using the maximum
possible transport. To capture the savings from local searches and
retaining the current solution, we instead count the work actually
performed.

\begin{theorem}[Maintenance work and queries]
\label{thm:main-guarantee}
After every update, \method{} maintains a set of density at least
$\rho^*/(1+\epsilon)$. Its total maintenance time is

\begin{equation}
 T_{\rm maint}=O\!\left(n+(R_{\rm up}+S_{\rm scan}+Z_{\rm norm})
 \log\frac{n+r m_*}{\epsilon}\right),
 \label{eq:readable-total-work}
\end{equation}

where $R_{\rm up}$ counts the endpoints of updated edges,
$S_{\rm scan}$ counts positive-incidence and endpoint inspections
in residual searches, and $Z_{\rm norm}$ counts the same inspections
when installing fallback witnesses.
Reading its density takes $O(1)$ time, and reporting its vertices takes
$\Theta(1+|C|)$ time per query.
\end{theorem}

The proof of \Cref{thm:amortized-update,thm:main-guarantee} is given in
\Cref{sec:maintenance-details}.
The input term $R_{\rm up}$ depends only on the updated edges.
The search term $S_{\rm scan}$ combines how much a search examines
(\Cref{sec:local-search-work}) with how often repair needs to move load
or replace the solution (\Cref{sec:reserve-transport,sec:total-maintenance-work}).
The normalization term $Z_{\rm norm}$ is paid only when a fallback
solution requires a new induced-edge count.
\Cref{sec:total-maintenance-work} bounds these three terms and combines
them into an amortized update bound.

\paragraph{Memory Analysis}
\method{} stores vertex loads, the indexed heap, witness membership,
and search marks and parent pointers in $O(n)$ words.
Each active incidence stores an allocation share and a constant number
of entries and inverse positions in the incidence lists. Together with
the endpoint arrays and active-edge records, these occupy $O(I)$ words,
where $I=\sum_{e\in E}|e|$ is the current incidence volume.
Searches and peeling use $O(n+I)$ temporary space.
Thus, a compacted representation uses $O(n+I)$ words in total.

\subsection{Bounding Search Work by Locality}
\label{sec:local-search-work}

We first bound $S_{\rm scan}$ in \Cref{eq:readable-total-work}.
A search expands only vertices whose load reaches its destination
threshold $t$. Since $q_v\le b\deg(v)$, let $h$ bound the number of
vertices with $b\deg(v)\ge t$ at the start of any search, and let $Q$
bound the maximum load at those times.
If no search occurs, set $h=Q=0$.
Let $A$ be the total amount transported and $F$ the number of failed
searches. A transfer of $a$ units contributes $a$ to $A$, once per path.

\begin{lemma}[Search work from local regions]
\label[lemma]{lem:local-search-work}
The total search work satisfies
\begin{equation}
 S_{\rm scan}\le(1+\min\{r,h\})hQ(A+F)+A.
 \label{eq:hyper-aggregate-search}
\end{equation}
\end{lemma}
The proof of \Cref{lem:local-search-work} is given in
\Cref{sec:local-search-proof}.
The factor $(1+\min\{r,h\})hQ$ bounds the work of one search,
apart from its final destination test. The factor $A+F$ bounds how
many searches occur. Thus a small high-load region reduces the work
per search; the next subsection bounds the transport that can cause
searches to repeat.

A large crossing edge need not cause a
full endpoint scan: only its prefix up to the first eligible
destination is tested. The bound still counts every inspected
positive incidence, including one whose edge was expanded earlier
in the same search.

For example, suppose only 20 vertices meet the degree threshold for
expansion. Every fully scanned edge then has at most 20 endpoints.
Even if the final edge has 1,000 endpoints, the search tests at most
21 of them before finding a destination. Thus the endpoint multiplier
is controlled by the small searchable region, not by the largest edge.
The factor $Q$ still accounts for how many positive shares these
vertices may own.

This situation can arise when a small, highly active community is
surrounded by many low-degree vertices. Once the maintained solution
sets a sufficiently high threshold, those low-degree vertices cannot
be expanded. The bound therefore explains how a large hypergraph can
admit inexpensive local searches; it does not require the whole graph
to be small.

\subsection{Bounding Transport by Density Changes}
\label{sec:reserve-transport}
\label{app:reserve-proof}

To make \Cref{lem:local-search-work} useful, we next bound $A$.
An insertion adds only $b$ units of load and never lowers the retained
solution density, so it transports at most $b$ units.
A deletion can lower the solution density and hence the capacities.
We charge its transport according to the density loss and the load
exposed to the lower threshold.

Write $s=\alpha bL$ and $\beta=\gamma/\alpha$.
For a deletion $d$, let $s_{\rm old}$ be the previous scale,
$\widetilde q$ the loads just after removing the edge, and $s_0$
the scale after witness normalization but before transport.
If the graph becomes empty or $s_0\ge s_{\rm old}$, set $Z_d=0$;
otherwise define
\begin{equation}
 Z_d=\frac{1+2/s_{\rm old}}{b}
       \left(\frac{s_{\rm old}}{s_0}-1\right)
       \sum_{\widetilde q_v+1>\beta s_0}\widetilde q_v.
 \label{eq:hyper-actual-scale-exposure}
\end{equation}
The ratio measures the relative drop in solution density.
The sum includes only load exposed to the new threshold, and division
by $b$ expresses this load in edge units. Thus a small density drop
affecting little load has a small exposure charge $Z_d$.

\begin{lemma}[Transport controlled by density loss]
\label[lemma]{lem:hyper-all-deletions}
For an update sequence starting from empty, with
$b=\lceil2r/\epsilon\rceil$ and $0<\epsilon\le1$, total transport satisfies
\begin{equation}
 A=O\!\left(\frac r\epsilon N_+
             +\frac r{\epsilon^2}\sum_d Z_d\right),
 \label{eq:exposure-transport}
\end{equation}
where $N_+$ is the number of insertions and $\sum_d Z_d$ is the total
exposure charge over all deletions in the update sequence.
\end{lemma}

The insertion term depends only on the new mass.
The deletion term depends on the actual density loss and exposed
load, rather than on all active edges at every deletion.
For example, deleting an edge outside the current solution gives
$Z_d=0$. If the solution contains $M$ edges and loses one while
retaining its vertices, its density falls by a fraction $1/M$.
For $M\ge2$, normalization can only improve this remaining density,
so $s_{\rm old}/s_0-1\le1/(M-1)$.
A solution with 100 edges therefore incurs a relative-scale charge
of at most $1/99$, whereas a two-edge solution can incur a charge of
one. The same single-edge deletion can thus have very different
effects on repair.

Only vertices whose load exceeds the shifted new threshold contribute
to the other factor in $Z_d$. Many low-load vertices do not increase this charge.
These two effects explain when aggregate transport can be small:
the dense solution persists through many updates, and its occasional
density losses expose little load. This is a plausible pattern in a
temporal interaction stream where individual interactions expire
while a well-connected group remains active.
We define a potential that assigns a quadratic cost to loads above
a shifted destination threshold. Each transfer decreases this
potential by an amount proportional to the mass transported.
Bounding the increases caused by insertions and density losses,
then summing over the update sequence, bounds total transport.
The full proof of \Cref{lem:hyper-all-deletions} is given in
\Cref{sec:transport-proof}.

\Cref{lem:transport-boundary} retains the boundary potentials for
a nonempty initial state. The degree-tail specialization in
\Cref{sec:degree-tail-specialization} further bounds $Z_d$ using
$\widetilde q_v\le b\deg_{\rm after}(v)$: only vertices of sufficiently
high degree can contribute exposed load.
In an $r$-uniform hypergraph, deleting the sole induced edge of a
normalized witness has zero exposure, as shown in
\Cref{lem:transport-boundary}.

Combining the transport bound with \Cref{lem:local-search-work}
now gives the search-work bound used in the running-time analysis.
\begin{lemma}[Search work under solution-guided maintenance]
\label[lemma]{lem:exposure-search-work}
Under the assumptions of \Cref{lem:hyper-all-deletions}, with $h,Q,F$
as in \Cref{lem:local-search-work},
\begin{equation}
\begin{aligned}
 S_{\rm scan}=O\!\Bigl(&(1+\min\{r,h\})hQ\\
 &{}\cdot\Bigl[F+\frac r\epsilon N_+
                     +\frac r{\epsilon^2}\sum_d Z_d\Bigr]\Bigr).
\end{aligned}
\label{eq:exposure-search-work}
\end{equation}
\end{lemma}
\begin{proof}[Proof of \Cref{lem:exposure-search-work}]
Substitute \Cref{eq:exposure-transport} in
\Cref{eq:hyper-aggregate-search}. If a search occurs, $h,Q\ge1$,
so the additive $A$ is absorbed by
$(1+\min\{r,h\})hQ(A+F)$. Otherwise $S_{\rm scan}=A=F=0$.
\end{proof}

The bound separates two sources of savings: a small high-load region
limits each search, and limited density loss reduces the repeated
transport caused by deletions. The failed-search count $F$ is
bounded next.

\subsection{From Repair Work to Running Time}
\label{sec:total-maintenance-work}

We now combine the search-work bound with the cost of applying updates
and normalizing fallback solutions.
Let $N_{\rm fb}$ count updates that install a fallback solution with
an unknown induced-edge count.

\begin{lemma}[Combined amortized update time]
\label[lemma]{lem:combined-update-time}
For the update sequence of \Cref{lem:exposure-search-work}, the amortized
update time, excluding initialization, is
\begin{equation}
\begin{aligned}
 O\!\Biggl(\Biggl[&r+\frac{(1+\min\{r,h\})hQ}{N}\\
 &{}\cdot\Bigl(F+\frac r\epsilon N_+
                    +\frac r{\epsilon^2}\sum_d Z_d\Bigr)\\
 &{}+I_*\frac{N_{\rm fb}}{N}\Biggr]
 \log\frac{n+rm_*}{\epsilon}\Biggr).
\end{aligned}
\label{eq:refined-amortized-work}
\end{equation}
\end{lemma}
\begin{proof}[Proof of \Cref{lem:combined-update-time}]
Applying an update inspects its endpoints, so
$R_{\rm up}=\sum_{\rm updates}|e|\le rN$.
\Cref{lem:exposure-search-work} bounds $S_{\rm scan}$.
Each fallback scans a positive incidence at most once and stamps each
encountered edge before testing its endpoints. Its incidence and
endpoint inspection counts are therefore each at most $I_*$, giving
$Z_{\rm norm}\le2I_*N_{\rm fb}$.
There is at most one such fallback per update; failed searches and
peeling already supply their induced-edge counts.
Substitute these three bounds in \Cref{eq:readable-total-work},
which includes path transfers, heap maintenance, and budgeted peeling.
Excluding initialization and dividing by $N$ proves the result.
\end{proof}

The expression has three parts: the size of the updated edge, local
search work averaged over the sequence, and fallback cost multiplied
by how often it occurs. Its middle term combines the two preceding
lemmas: a small high-load region lowers the cost per search, while
limited density loss restricts how much transport needs to repeat.
The last term is small when the maintained solution rarely needs
fallback normalization.

Failed searches also make progress: each raises the inner threshold
by a factor of at least $1+\epsilon/4$
(\Cref{eq:failure-threshold-growth}), giving
$O(\epsilon^{-1}\log(2+rm_*))$ failures per update.
The sequence-wide bound in \Cref{lem:search-counts} further charges
repeated threshold increases to decreases caused by deletions.

For example, deleting an edge outside the current witness leaves its
density and capacity unchanged and can only lower loads. Such an update
has no search or normalization cost, so its maintenance time is
$O(|e|\log((n+rm_*)/\epsilon))$.
The counter bound in \Cref{eq:readable-total-work} captures this case
directly, while \Cref{eq:refined-amortized-work} bounds the contributions
over the full sequence.

More generally, a stream can be large while most changes preserve a
dense community and the repairs remain confined to a small surrounding
region. Occasional changes within that community can also be inexpensive
when it contains many edges, as the density-loss example in
\Cref{sec:reserve-transport} illustrates.
These patterns offer a plausible explanation for the observed
efficiency of solution-guided maintenance.
The running-time results (\Cref{sec:main-results}) and the reductions
in positive scans in the capacity-slack and reserve-destination
ablation (\Cref{sec:ablation}) are consistent with this explanation
of the observed efficiency.

%% file: sections/evaluation.tex
\section{Evaluation}
\label{sec:evaluation}
We evaluate the efficiency and solution quality of \texttt{CAP} on real-world
hypergraph datasets and ordinary graphs ($r=2$). We first introduce the
datasets, baselines, and query schedules (\Cref{sec:experiment-setting}),
and compare running times on dynamic hypergraphs
(\Cref{sec:main-results}). We then validate the approximation quality of the
returned solutions (\Cref{sec:validation-records}), study scalability with
query frequency (\Cref{sec:query-frequency}) and data size
(\Cref{sec:data-size}), examine the effect of the
accuracy parameter (\Cref{sec:precision}),
evaluate witness-guided repair through an ablation study (\Cref{sec:ablation}),
and compare performance on ordinary graphs (\Cref{sec:ordinary-graphs}).

\subsection{Experiment Setting}
\label{sec:experiment-setting}
\paragraph{Datasets.}
We select five publicly-available real-world hypergraph datasets from the
Cornell collection~\cite{Benson-2018-simplicial}\footnote{\url{https://www.cs.cornell.edu/~arb/data/}}:
\texttt{DW}, \texttt{MA}, \texttt{AU}, \texttt{GE}, and \texttt{SO}. \texttt{DW} records drug
co-occurrences; \texttt{MA}, \texttt{AU}, and \texttt{SO} contain tags assigned to
online questions; \texttt{GE} records publication coauthorship.
\Cref{tab:calendar-data} shows the statistics of these hypergraph data.
We remove hyperedges containing only one vertex and repeated hyperedges
that are already active, same as in \texttt{Udshp}~\cite{hd_bera2022}.
We use the same calendar-based query schedule as
\texttt{Udshp}~\cite{hd_bera2022}.
Specifically, the update operations are generated in the following manners:
\begin{itemize}[leftmargin=*]
  \item \textbf{Quarter-based windows.} \texttt{DW} provides quarterly timestamps.
  We insert hyperedges in temporal order, delete those from the preceding
  quarter, and issue queries at quarter boundaries.
  \item \textbf{Timestamp-based windows.} \texttt{MA}, \texttt{AU}, and \texttt{SO}
  provide event timestamps. We normalize them into 90-day bins, retain bin
  indices greater than one, and maintain a one-bin window. Incoming tag sets
  generate insertions, expired hyperedges generate deletions, and queries
  occur at bin boundaries, including bins without new events.
  \item \textbf{Year-based windows.} \texttt{GE} provides publication years.
  We maintain a ten-year window, insert incoming publications, delete expired
  hyperedges, and issue annual queries, including years without new
  publications.
\end{itemize}
These calendar boundaries yield 32, 29, 33, 210, and 35 queries on
\texttt{DW}, \texttt{MA}, \texttt{AU}, \texttt{GE}, and \texttt{SO}, respectively.
Queries at empty calendar bins are retained, so several queries can follow
the same update.

For the query-frequency experiment, we keep the update stream fixed and
increase the original number of queries, $Q_0$, to $4Q_0$ and $16Q_0$.
We divide the update sequence into $15Q_0$ equal-length intervals and add
one query after the midpoint update of each interval. If that position is
already queried, we select the nearest unused update position, breaking
ties toward the earlier update. The $16Q_0$ schedule contains all these
queries and all original queries. The $4Q_0$ schedule retains the middle
added query from each consecutive group of five, together with all original
queries. We use 32, 128, and 512 queries on \texttt{DW}, and 210, 840,
and 3,360 queries on \texttt{GE}.

\input{sections/calendar_data}

We also use six ordinary graph datasets: \texttt{EM}, \texttt{HT}, \texttt{WV}, \texttt{BC},
\texttt{AS}, and \texttt{SW}.\footnote{\url{https://snap.stanford.edu/data/};
the dynamic \texttt{AS} and \texttt{SW} streams are available from
\url{https://dyreach.taa.univie.ac.at/}.}
For \texttt{EM}, \texttt{HT}, and \texttt{WV}, we initialize half of the edges and generate
20,000 random edge insertions or deletions on the source topology.
For \texttt{BC}, we process ratings in timestamp order using a 90-day window.
\texttt{AS} uses changes between consecutive snapshots, and \texttt{SW} uses
timestamped hyperlink insertions and deletions. We project these data to
undirected, unweighted graphs; in \texttt{SW}, we exclude pairs with an
unmatched deletion in their chronological history.
The last six rows of \Cref{tab:calendar-data} summarize these graphs.

\paragraph{Baselines.}
The comparison includes \texttt{CAP} and the following five baselines.
\begin{itemize}[leftmargin=*]
  \item \texttt{Udshp}~\cite{hd_bera2022}: Our primary dynamic baseline
  directly addresses the same densest-subhypergraph objective under
  hyperedge insertions and deletions. To the best of our knowledge, the
  authors' implementation is the most recent publicly available
  implementation for this problem. It maintains a hypergraph orientation
  and identifies a dense set through local endpoint-load conditions.
  \item \rev{\texttt{CQ}~\cite{chekuri2024}: This theoretical
  paper studies ordinary graphs in its main results. A preliminary version
  discusses a hypergraph extension in Section~6, omitting the deletion
  pseudocode.\footnote{\url{https://arxiv.org/abs/2210.02611}}}
  We could not locate an accompanying implementation, so we implement
  its basic amortized core and hypergraph extension in C++.
  Our version uses unit-copy updates, shared multiplicities, nested label
  buckets, and load-threshold extraction. We complete hypergraph deletion
  by selecting a global minimum-load endpoint and processing changed
  endpoints through a worklist. Implementation details and scope are
  given in \Cref{app:cq-reconstruction}.
  \item \texttt{Exact}~\cite{hd_bera2022}: This linear-programming algorithm
  solves the current hypergraph using GLOP and extracts a
  dense set from the positive-support vertices.
  \item \texttt{DI}~\cite{hd_huang2024hyper}: Density Improvement repeatedly
  solves flow subproblems to increase the density. We use its exact
  global solver.
  \item \texttt{ExactFlow}~\cite{hd_bengali2026}: This flow algorithm
  maximizes density under convex, monotonic hyperedge rewards. We use the
  reward vector $[0,\ldots,0,1]$ for fully contained-edge density.
\end{itemize}
We exclude \texttt{HWC}~\cite{hd_hu2017} from the baselines mainly because
its rank-dependent $r^2(1+\epsilon_H)$ guarantee does not ensure the
$(1+\epsilon)$ approximation required in our comparison.
We also considered an adaptation of \texttt{IDH}~\cite{hd_leng2026decomposition},
a dynamic integer-density decomposition method. Our adaptation uses
$\delta=1$, modified load-based path-reversal maintenance, and extraction
of the highest-layer vertex set, with
$b=\lceil(1+\epsilon)r/\epsilon\rceil$ explicit copies per hyperedge to meet
our approximation target. Even at $\epsilon=1$, this variant did not
complete the full update-and-query stream on \texttt{AU}, \texttt{MA}, or
\texttt{DW} within a 300-second wall-time limit in our single-run pilots.
We therefore exclude it from the baseline comparison.

\paragraph{Setting.}
We run experiments on a machine having an Intel Core i9-10850K CPU and
32 GB of memory, with Windows 11 installed. C++ implementations are compiled
with GCC 12.1 and optimization enabled. The two static flow solvers use
Julia 1.9.4 with one thread, and \texttt{Exact} uses OR-Tools 9.11.4210.
We run experiments serially and report mean total running times only
after three complete, independently validated repetitions. Times are rounded
to two decimal places, and the smallest reported mean is shown in bold.
For hypergraph queries returning full vertex sets, total running time
covers maintenance and queries,
including static-solver input construction and materialization of returned
vertices. It excludes input parsing, initialization, and independent
verification; Julia warmup precedes timing. The main,
query-frequency, and accuracy-parameter comparisons use a 1,800-second
whole-process limit, while ablation runs use a 300-second limit.

\paragraph{Parameters.}
For a given $\epsilon$, \texttt{CAP} targets density
$\rho^*/(1+\epsilon)$ and uses reserve destinations,
$b=\lceil2r/\epsilon\rceil$, and budgeted peeling.
For \texttt{Udshp}, we use the same numerical $\epsilon$ as its native
experimental parameter, with copy count $\operatorname{round}(\epsilon^{-2})$
as in the released implementation~\cite{hd_bera2022}.
For \texttt{CQ}, we derive fixed copy and local-balance parameters
from the requested target $\rho^*/(1+\epsilon)$
(\Cref{app:cq-reconstruction}).
The three static methods solve each query snapshot exactly from scratch and have no
approximation parameter. We therefore report one mean per dataset, shared
across the five $\epsilon$ rows. We set \texttt{DI}'s numerical flow tolerance
to $10^{-8}$ and \texttt{ExactFlow}'s to $10^{-6}$.

\subsection{Efficiency on Dynamic Hypergraphs}
\label{sec:main-results}
We compare running times at $\epsilon\in\{1,0.7,0.5,0.2,0.1\}$.
\input{sections/unified_results}
\texttt{CAP} has the lowest running time in all 25 settings in
\Cref{tab:unified-time}, while meeting the requested approximation
after every update. Against the dynamic baselines, it is
\MainSpeedup{}$\times$ faster than \texttt{Udshp} in the
\NativeComparisonCells{} completed native-parameter comparisons.
Against \texttt{CQ}, the four completed comparisons yield
$2343.5$--$5926.6\times$ speedups.
At $\epsilon=0.1$, both dynamic baselines have TL entries on every
dataset, whereas \texttt{CAP} completes all five streams.
Tightening $\epsilon$ from 1 to $0.1$ raises \texttt{CAP}'s time by
only $1.5$--$5.9\times$ across these datasets.

The advantage also holds against exact snapshot recomputation, despite
the relatively sparse calendar queries. At $\epsilon=0.1$,
\texttt{CAP} is $3.4$--$259.6\times$ faster than \texttt{DI} and
$5.4$--$246.9\times$ faster than \texttt{ExactFlow}.
The largest gains occur on \texttt{GE}, where repeatedly solving the
query snapshots is expensive. Even on \texttt{DW}, which has only 32
queries, \texttt{CAP} is $41.2\times$ faster than \texttt{Exact}.
Retaining a valid solution therefore benefits both update-heavy streams
and streams with expensive query snapshots.

\subsection{Quality Validation}
\label{sec:validation-records}
\input{sections/validation_records}

\input{sections/query_frequency}

\input{sections/size_scalability}

\subsection{Effect of the Accuracy Parameter}
\label{sec:precision}
We vary $\epsilon$ over $1$, $0.7$, $0.5$, $0.2$, $0.1$, $0.08$,
$0.06$, $0.04$, $0.02$, and $0.01$ on \texttt{GE}, the dataset with
the most vertices.
\input{sections/geology_precision}
\Cref{fig:geology-precision} shows that tightening accuracy by two
orders of magnitude increases \texttt{CAP}'s time by only $2.6\times$.
At $\epsilon=0.01$, it completes in 6.36 seconds and its minimum
density ratio exceeds the required $1/1.01$.
In contrast, \texttt{Udshp} has TL entries from $\epsilon=0.1$
onward, and \texttt{CQ} already has TL at $\epsilon=1$.
Even at the strictest setting, \texttt{CAP} remains $151.0\times$
and $143.5\times$ faster than \texttt{DI} and \texttt{ExactFlow},
respectively. Thus, high accuracy preserves a substantial advantage
on this large hypergraph without requiring a proportional increase
in running time.

This robustness is dataset dependent: at $\epsilon=0.01$,
\texttt{DI} is faster on \texttt{DW}, \texttt{MA}, and \texttt{AU},
and \texttt{ExactFlow} is faster on \texttt{MA} and \texttt{AU}.

\input{sections/ablation}

\subsection{Results on Ordinary Graphs}
\label{sec:ordinary-graphs}
We also evaluate \texttt{CAP} on six ordinary graph streams with
$\epsilon=0.1$. We compare it with
\texttt{ImprovedDynOpt}~\cite{grossmann2025exact} and
\texttt{INS/DEL}~\cite{zhang2024pseudo}, which maintain an optimal graph
orientation and pseudoarboricity, respectively.
We sum the recorded running times over the complete update stream, including
the initial edge insertions. The measured cost includes \texttt{CAP}'s updates
and the baselines' updates plus scalar queries every 100 operations;
\texttt{CAP}'s scalar reads and vertex enumeration are untimed.
We check solution quality at sampled queries.
\input{sections/ordinary_graphs}
\input{sections/ordinary_discussion}

%% file: sections/calendar_data.tex
\begin{table}[t]
\centering\small
\setlength{\tabcolsep}{1.5pt}
\caption{Datasets used in experiments. $n$: number of vertices; $r$: maximum edge size; $m_{\max}$: peak number of active edges.}
\label{tab:calendar-data}
\begin{tabular}{lrrrrr}
\toprule
Dataset (Abbr.) & $n$ & $r$ & $m_{\max}$ & Updates & Queries\\
\midrule
\texttt{DAWN (DW)} & 2,290 & 16 & 11,544 & 529,412 & 32\\
\texttt{Math (MA)} & 1,627 & 5 & 21,349 & 611,274 & 29\\
\texttt{AskUbuntu (AU)} & 3,021 & 5 & 10,458 & 382,007 & 33\\
\texttt{Geology (GE)} & 1,087,111 & 25 & 409,152 & 1,471,530 & 210\\
\texttt{StackOverflow (SO)} & 49,918 & 5 & 351,521 & 15,897,567 & 35\\
\midrule
\texttt{Email (EM)} & 986 & 2 & 8,058 & 28,032 & 281\\
\texttt{HepTh (HT)} & 9,875 & 2 & 13,040 & 32,986 & 330\\
\texttt{WikiVote (WV)} & 7,115 & 2 & 50,564 & 70,381 & 704\\
\texttt{Bitcoin (BC)} & 5,881 & 2 & 2,812 & 44,434 & 445\\
\texttt{AS-CAIDA (AS)} & 31,379 & 2 & 53,613 & 587,677 & 5,877\\
\texttt{SimpleWiki (SW)} & 100,312 & 2 & 698,216 & 1,439,148 & 14,392\\
\bottomrule
\end{tabular}
\end{table}

%% file: sections/unified_results.tex
\begin{table}[t]
\centering\small
\setlength{\tabcolsep}{2pt}
\renewcommand{\arraystretch}{1.10}
\caption{Average running time (s). TL: time limit.}
\label{tab:unified-time}
\begin{tabular*}{\columnwidth}{@{\extracolsep{\fill}}cccccccc@{}}
\toprule
\multirow{2}{*}{Data} & \multirow{2}{*}{$\epsilon$} & \multicolumn{3}{c}{Dynamic methods} & \multicolumn{3}{c}{Static methods}\\
\cmidrule(lr){3-5}\cmidrule(l){6-8}
 & & \texttt{CAP} & \texttt{Udshp} & \texttt{CQ} & \texttt{Exact} & \texttt{DI} & \texttt{ExactFlow}\\
\midrule
\multirow{5}{*}{\texttt{DW}} & 1 & \textbf{0.24} & 13.52 & 1448.21 & \multirow{5}{*}{34.72} & \multirow{5}{*}{4.69} & \multirow{5}{*}{8.93}\\
 & 0.7 & \textbf{0.32} & 46.19 & TL &  &  & \\
 & 0.5 & \textbf{0.35} & 168.60 & TL &  &  & \\
 & 0.2 & \textbf{0.55} & TL & TL &  &  & \\
 & 0.1 & \textbf{0.84} & TL & TL &  &  & \\
\midrule
\multirow{5}{*}{\texttt{MA}} & 1 & \textbf{0.28} & 14.22 & 824.18 & \multirow{5}{*}{54.19} & \multirow{5}{*}{5.93} & \multirow{5}{*}{10.34}\\
 & 0.7 & \textbf{0.33} & 47.27 & TL &  &  & \\
 & 0.5 & \textbf{0.37} & 159.40 & TL &  &  & \\
 & 0.2 & \textbf{0.69} & TL & TL &  &  & \\
 & 0.1 & \textbf{1.59} & TL & TL &  &  & \\
\midrule
\multirow{5}{*}{\texttt{AU}} & 1 & \textbf{0.21} & 4.03 & 500.08 & \multirow{5}{*}{20.35} & \multirow{5}{*}{4.26} & \multirow{5}{*}{6.76}\\
 & 0.7 & \textbf{0.25} & 11.76 & 1112.02 &  &  & \\
 & 0.5 & \textbf{0.29} & 38.49 & TL &  &  & \\
 & 0.2 & \textbf{0.58} & 1079.61 & TL &  &  & \\
 & 0.1 & \textbf{1.25} & TL & TL &  &  & \\
\midrule
\multirow{5}{*}{\texttt{GE}} & 1 & \textbf{2.47} & 38.34 & TL & \multirow{5}{*}{TL} & \multirow{5}{*}{959.86} & \multirow{5}{*}{912.67}\\
 & 0.7 & \textbf{2.43} & 57.40 & TL &  &  & \\
 & 0.5 & \textbf{2.48} & 103.25 & TL &  &  & \\
 & 0.2 & \textbf{3.10} & 935.33 & TL &  &  & \\
 & 0.1 & \textbf{3.70} & TL & TL &  &  & \\
\midrule
\multirow{5}{*}{\texttt{SO}} & 1 & \textbf{11.94} & 1680.65 & TL & \multirow{5}{*}{TL} & \multirow{5}{*}{320.04} & \multirow{5}{*}{484.23}\\
 & 0.7 & \textbf{13.22} & TL & TL &  &  & \\
 & 0.5 & \textbf{12.52} & TL & TL &  &  & \\
 & 0.2 & \textbf{18.90} & TL & TL &  &  & \\
 & 0.1 & \textbf{36.44} & TL & TL &  &  & \\
\bottomrule
\end{tabular*}
\end{table}

%% file: sections/validation_records.tex
We check that the speedups preserve \texttt{CAP}'s required approximation.
An independent integer-allocation observer tracks every update and checks
edge mass, vertex loads, and the maintained solution density against
\[
 L\le\rho^*\le U\le(1+\epsilon)L,
 \qquad U=\max_v q_v/b.
\]
Its ledger is fully rescanned every 1,024 updates and at queries.
All original-calendar validation trajectories pass.
\Cref{tab:cap-quality} further compares returned densities with exact
maximum-closure optima at every calendar query; all tested
dataset/accuracy combinations meet $1/(1+\epsilon)$.
The additional queries in \Cref{sec:query-frequency} also meet their
requested target. Thus, the runtime advantage is accompanied by
validated solution quality, including between scheduled queries.
\input{sections/five_quality}

We independently check every returned set from the complete
\texttt{CQ} runs against exact snapshot optima; all meet the requested
density ratio $1/(1+\epsilon)$.

%% file: sections/five_quality.tex
\begin{table}[t]
\centering\small
\setlength{\tabcolsep}{3pt}
\renewcommand{\arraystretch}{1.12}
\caption{Minimum \texttt{CAP} density divided by the exact optimum across complete validation trajectories. Column headings give $\epsilon$.}
\label{tab:cap-quality}
\begin{tabular*}{\columnwidth}{@{\extracolsep{\fill}}ccccccc@{}}
\toprule
Data & 1 & 0.7 & 0.5 & 0.2 & 0.1 & 0.01\\
\midrule
\texttt{DW} & 0.5846 & 0.6099 & 0.7150 & 0.8565 & 0.9324 & 0.9964\\
\texttt{MA} & 0.5299 & 0.6223 & 0.7215 & 0.8776 & 0.9528 & 0.9961\\
\texttt{AU} & 0.5475 & 0.6545 & 0.7012 & 0.9087 & 0.9483 & 0.9926\\
\texttt{GE} & 0.5294 & 0.6120 & 0.6667 & 0.8333 & 0.9184 & 0.9913\\
\texttt{SO} & 0.5627 & 0.6119 & 0.6987 & 0.8664 & 0.9394 & 0.9977\\
\bottomrule
\end{tabular*}
\end{table}

%% file: sections/query_frequency.tex
\Needspace{6\baselineskip}
\subsection{Scalability with Query Frequency}
\label{sec:query-frequency}
On \texttt{DW} and \texttt{GE}, we fix $\epsilon=0.5$ and use the
query schedules in \Cref{sec:experiment-setting}.
\Cref{tab:query-frequency} shows the results.
\input{sections/query_frequency_table}
\input{sections/every_update_size_table}

Increasing query frequency sixteenfold leaves \texttt{CAP}'s running
time almost unchanged on both datasets. The maintained solution makes
each query an output operation, avoiding a new optimization or
load-threshold search. With 512 queries on \texttt{DW}, \texttt{CAP}
is $113.9$--$671.1\times$ faster than the three static methods and
$551.9\times$ faster than \texttt{Udshp} under its native parameter.
On \texttt{GE}, only \texttt{CAP} completes at the highest frequency;
the \texttt{Udshp}, \texttt{CQ}, and static-solver entries are all TL.
These results show that explicitly maintaining the solution becomes
particularly useful when queries are frequent.

\paragraph{Reporting after every update.}
We further evaluate \texttt{CAP} on all five hypergraph datasets with
$\epsilon\in\{0.1,0.01\}$.
We keep every insertion and deletion unchanged
and issue exactly one query immediately after each update. Each query
returns only the current solution size $|C|$, using the maintained
cardinality without enumerating its vertices. Thus, the number of queries
equals the number of updates. \Cref{tab:every-update-size} reports the
total running time for maintenance and these size-only responses.

At $\epsilon=0.1$, even the largest
stream, \texttt{SO}, takes only 37.27 seconds for 15.9 million
updates and responses. Stricter accuracy can make maintenance more
expensive: at $\epsilon=0.01$, the same stream takes 1586.63 seconds,
whereas \texttt{GE} still finishes in 6.48 seconds.
The constant-time size query removes enumeration cost; it does not
remove the dataset-dependent cost of maintaining a more accurate solution.

%% file: sections/query_frequency_table.tex
\begin{table}[t]
\centering\small
\setlength{\tabcolsep}{2pt}
\renewcommand{\arraystretch}{1.15}
\caption{Average running time (s) with varying query counts. TL: time limit.}
\label{tab:query-frequency}
\begin{tabular*}{\columnwidth}{@{\extracolsep{\fill}}cccccccc@{}}
\toprule
Data & Queries & \texttt{CAP} & \texttt{Udshp} & \texttt{CQ} & \texttt{Exact} & \texttt{DI} & \texttt{ExactFlow}\\
\midrule
\multirow{3}{*}{\texttt{DW}} & 32 & \textbf{0.35} & 174.39 & TL & 33.37 & 4.62 & 8.86\\
 & 128 & \textbf{0.34} & 172.94 & TL & 72.92 & 11.42 & 22.53\\
 & 512 & \textbf{0.35} & 190.42 & TL & 231.53 & 39.30 & 74.20\\
\midrule
\multirow{3}{*}{\texttt{GE}} & 210 & \textbf{2.58} & 103.25 & TL & TL & 959.86 & 912.67\\
 & 840 & \textbf{2.53} & 1107.58 & TL & TL & TL & TL\\
 & 3,360 & \textbf{2.51} & TL & TL & TL & TL & TL\\
\bottomrule
\end{tabular*}
\end{table}

%% file: sections/every_update_size_table.tex
\begin{table}[!t]
\centering\small
\renewcommand{\arraystretch}{1.15}
\caption{Average running time with a size-only query after every update.}
\label{tab:every-update-size}
\begin{tabular*}{\columnwidth}{@{\extracolsep{\fill}}cccc@{}}
\toprule
\multirow{2}{*}{Data} & \multirow{2}{*}{Updates / queries} & \multicolumn{2}{c}{\texttt{CAP} time (s)}\\
\cmidrule(lr){3-4}
 & & $\epsilon=0.1$ & $\epsilon=0.01$\\
\midrule
\texttt{DW} & 529,412 & 0.86 & 4.97\\
\texttt{MA} & 611,274 & 1.60 & 22.85\\
\texttt{AU} & 382,007 & 1.26 & 14.15\\
\texttt{GE} & 1,471,530 & 3.80 & 6.48\\
\texttt{SO} & 15,897,567 & 37.27 & 1586.63\\
\bottomrule
\end{tabular*}
\end{table}

%% file: sections/size_scalability.tex
\Needspace{6\baselineskip}
\subsection{Scalability with Data Size}
\label{sec:data-size}
We evaluate \texttt{CAP} on \texttt{GE} and \texttt{SO}, the two datasets
with the largest peak numbers of active hyperedges, at $\epsilon=0.1$.
For each dataset, we uniformly shuffle the vertex set with a fixed
random seed and retain its first $\lfloor k|V|/100\rfloor$ vertices
for $k\in\{20,40,60,80,100\}$. Each nested sample induces a stream
containing exactly the original hyperedges whose endpoints all belong
to the sample. We preserve their insertion/deletion order, the full
calendar query schedule (including empty snapshots), and the original
rank upper bound. Queries return full vertex sets. The 100\% points
reuse the complete measurements from \Cref{sec:main-results}.

\begin{figure}[t]
\centering
\includegraphics[width=0.78\columnwidth]{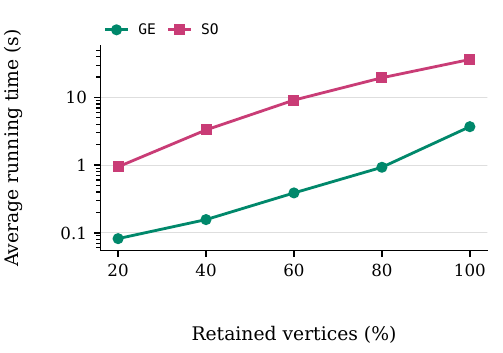}
\caption{Average running time of \texttt{CAP} with varying data size
($\epsilon=0.1$). The vertical axis is logarithmic.}
\Description{CAP running times on GE and SO at 20, 40, 60, 80, and 100 percent
vertex retention. Both curves increase across these nested samples.}
\label{fig:data-size}
\end{figure}

\Cref{fig:data-size} shows increasing running time as the samples grow,
with \texttt{CAP} completing even the full \texttt{SO} stream in
36.44 seconds. From 20\% to 100\% vertex retention on \texttt{SO},
the update count grows $96.2\times$ but running time grows only
$38.7\times$.
This behavior is consistent with maintenance that often processes
updates locally, rather than rescanning the entire active hypergraph.
Vertex-induced sampling changes both the edge set and its rank
distribution, so the horizontal axis is a sampling fraction rather
than a direct measure of input volume.

%% file: sections/geology_precision.tex
\begin{figure}[t]
\centering
\includegraphics[width=0.78\columnwidth]{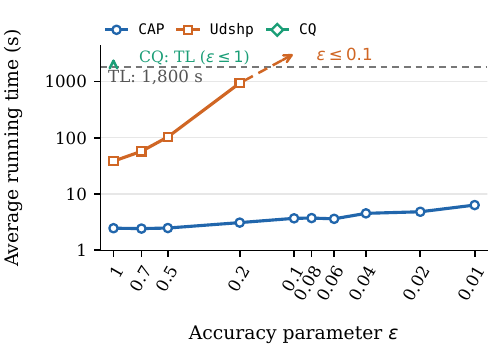}
\caption{Average running time with varying precision on \texttt{GE}.
Arrows mark the 1,800-second whole-process limit.
}
\Description{CAP and Udshp running times with the CQ reconstruction added to the comparison. }
\label{fig:geology-precision}
\end{figure}

%% file: sections/ablation.tex
\Needspace{6\baselineskip}
\subsection{Ablation Study}
\label{sec:ablation}
We study witness-guided repair
on \texttt{DW}, \texttt{MA}, \texttt{AU}, \texttt{GE}, and \texttt{SO} with $\epsilon=0.1$.
We compare \texttt{CAP}'s outer capacity and reserve destinations
(\Cref{sec:cap-method}) with a tighter variant:
\begin{itemize}[leftmargin=*]
  \item \texttt{CAP}: Outer capacity $\lfloor(1+\epsilon)bL\rfloor$;
  destinations have load below $\lfloor(1+\epsilon/2)bL\rfloor$.
  \item \texttt{CAP-Tight}: Capacity $\lceil bL\rceil$;
  any destination below that capacity is accepted.
\end{itemize}
All completed runs meet the required approximation at every query.
\input{sections/ablation_figure}
\Cref{fig:witness-ablation} shows the combined benefit of the larger
outer capacity and reserve destinations. \texttt{CAP} is
$2.0$--$57.9\times$ faster on the four datasets where both variants
complete, with $3.6$--$158.0\times$ fewer positive-incidence scans.
On \texttt{SO}, \texttt{CAP-Tight} exceeds 300 seconds, while
\texttt{CAP} completes. The lower scan counts are consistent with
reduced repair work from the combined use of the larger outer capacity
and reserve destinations.

%% file: sections/ablation_figure.tex
\begin{figure}[t]
\centering
\includegraphics[width=0.88\columnwidth]{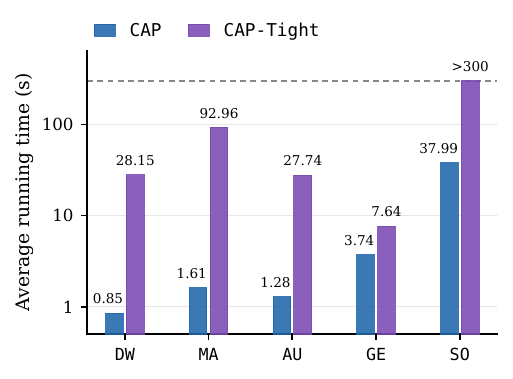}
\caption{Average running time in the ablation study ($\epsilon=0.1$).}
\Description{CAP and CAP-Tight running times on \texttt{DW}, \texttt{MA}, \texttt{AU}, \texttt{GE}, and \texttt{SO}.
The vertical axis uses a log scale.}
\label{fig:witness-ablation}
\end{figure}

%% file: sections/ordinary_graphs.tex
\begin{table}[t]
\centering\small
\setlength{\tabcolsep}{2pt}
\renewcommand{\arraystretch}{1.10}
\caption{Average total running time (s) on ordinary graphs. Speedup is baseline time divided by \texttt{CAP} time.}
\label{tab:ordinary-graphs}
\begin{tabular*}{\columnwidth}{@{\extracolsep{\fill}}lcccccc@{}}
\toprule
Algorithm & \texttt{EM} & \texttt{HT} & \texttt{WV} & \texttt{BC} & \texttt{AS} & \texttt{SW}\\
\midrule
\multicolumn{7}{c}{Total time (s)}\\
\texttt{CAP} & \textbf{0.03} & \textbf{0.02} & \textbf{0.19} & \textbf{0.11} & 1.61 & 5.24\\
\texttt{ImprovedDynOpt} & 0.08 & 0.68 & 0.78 & 0.30 & \textbf{1.52} & \textbf{4.53}\\
\texttt{INS/DEL} & 0.73 & 0.33 & 11.08 & 0.15 & 6.04 & 334.62\\
\midrule
\multicolumn{7}{c}{\texttt{CAP} speedup over each baseline}\\
\texttt{ImprovedDynOpt} & $2.2\times$ & $29.3\times$ & $4.1\times$ & $2.7\times$ & $0.9\times$ & $0.9\times$\\
\texttt{INS/DEL} & $21.1\times$ & $14.2\times$ & $57.9\times$ & $1.4\times$ & $3.7\times$ & $63.8\times$\\
\bottomrule
\end{tabular*}
\end{table}

%% file: sections/ordinary_discussion.tex
\Cref{tab:ordinary-graphs} shows that \texttt{CAP}'s benefit extends to
ordinary graphs: it is $1.4$--$63.8\times$ faster than \texttt{INS/DEL}
on all six streams and $2.2$--$29.3\times$ faster than
\texttt{ImprovedDynOpt} on four. The advantage is not universal:
\texttt{ImprovedDynOpt} is modestly faster on \texttt{AS} and
\texttt{SW}. Thus, \texttt{CAP} remains competitive with specialized
graph implementations while supporting general hyperedges.

%% file: appendix.tex
\appendix
\crefalias{section}{appendix}
\crefalias{subsection}{appendix}
\numberwithin{equation}{section}
\input{sections/analysis_appendix}
\input{sections/cq_reconstruction}

%% file: sections/analysis_appendix.tex
\section{Transport Bounds and Supporting Derivations}
\label{sec:transport-details}
We collect the proofs of the search-work and transport bounds,
using the notation of \Cref{sec:local-search-work,sec:reserve-transport}.
We also give the supporting update-charge inequalities and
the degree-tail bound.

\subsection{Proof of the Local Search-Work Bound}
\label{sec:local-search-proof}

\begin{proof}[Proof of \Cref{lem:local-search-work}]
Consider one search and set $H=\{v:b\deg(v)\ge t\}$, so $|H|\le h$.
Let $P$ count inspected positive incidences, $E$ distinct expanded
hyperedges, and $J$ endpoint tests. All expanded vertices have load at
least $t$ and belong to $H$. Each positive incidence carries at least
one integer unit, so their total number is bounded by the sum of
loads over $H$. Also, every expanded edge consumes an inspected
incidence. Therefore $P\le hQ$ and $E\le P$.

A fully expanded edge contains only vertices of load at least $t$:
otherwise the search would have stopped at an eligible destination.
Such an edge lies in $H$ and has at most $\min\{r,h\}$ endpoints.
The last edge of a successful search may leave $H$, but scanning it
stops at the first lower-load endpoint. If $\sigma$ indicates success,
these observations give
\begin{align*}
 J&\le\min\{r,h\}E+\sigma,\\
 P+J&\le(1+\min\{r,h\})hQ+\sigma.
\end{align*}
Let $K$ be the number of successful searches. Summing over all $K+F$
searches yields
\[
 S_{\rm scan}\le(1+\min\{r,h\})hQ(K+F)+K.
\]
Each successful transfer moves a positive integer amount, so $K\le A$,
which proves \Cref{eq:hyper-aggregate-search}.
\end{proof}

\subsection{Proof of the Transport Bound}
\label{sec:transport-proof}
\subsubsection{Potential Decrease under a Transfer}
\label{sec:potential-exposure}

Set $\delta=\epsilon/\alpha$, so that $c=\lfloor s\rfloor$,
$t=\lfloor\beta s\rfloor$, and $1-\beta=\delta/2$.
We use the potential
\[
 \Phi_s(q)=\frac1s\sum_v(q_v-\beta s+1)_+^2,
 \qquad (x)_+=\max\{x,0\},
\]
and assign potential zero to an empty state.
It charges only loads above $\beta s-1$, quadratically in their excess.
An eligible destination has $q_w<t$ and contributes zero before
the transfer. The shift by one handles integer rounding.

\begin{lemma}[Potential decrease of a reserve transfer]
\label[lemma]{lem:reserve-potential-drop}
At a fixed scale $s$, a successful reserve path of amount $a$ satisfies
\begin{equation}
 \Phi_s(q^{\rm before})-\Phi_s(q^{\rm after})\ge\delta a.
\label{eq:reserve-path-potential}
\end{equation}
\end{lemma}
\begin{proof}[Proof of \Cref{lem:reserve-potential-drop}]
Write $a_s=\beta s-1$. Only the root $z$ and destination $w$ change their loads. An eligible
integer destination satisfies $q_w\le t-1\le a_s$, so its initial
contribution is zero and its final unscaled contribution is at most $a^2$.
By \Cref{eq:cap-transfer}, the root ends at $q_z-a\ge c$.
Since $c-a_s\ge(s-1)-(\beta s-1)=(1-\beta)s>0$, its decrease before
division by $s$ is
\begin{align*}
 (q_z-a_s)^2-(q_z-a-a_s)^2
 &=2a(q_z-a-a_s)+a^2\\
 &\ge2a(c-a_s)+a^2.
\end{align*}
Subtracting the destination's increase and dividing by $s$ leaves at
least $2a(1-\beta)=\delta a$.
\end{proof}

Each unit transported thus releases at least $\delta$ units of
potential. This converts a sequence of repairs into an accounting of
the potential introduced by updates. \rev{Convex load potentials have
precedent~\cite{christiansen2025local}}; here the scale follows
the maintained solution and can decrease after a deletion.

\subsubsection{Properties of the Shifted Potential}
Write $a_s=\beta s-1$. From $bL\ge2/\epsilon$, we obtain
\[
 s\ge\frac{2(1+\epsilon)}{\epsilon}=\frac2\delta,
 \qquad a_s=(1+\epsilon/2)bL-1\ge\frac2\epsilon.
\]
In particular, $a_s\ge0$ and $1/s\le\delta/2$.
For $q>a_s$, the derivative of $(q-a_s)^2/q$ is
$1-a_s^2/q^2\ge0$. For $q\le a_s$, its positive-part version is zero.
Consequently, for every stable load $0<q\le c\le s$,
\[
 \frac{(q-a_s)_+^2}{sq}
 \le\frac{(s-a_s)^2}{s^2}
 =(1-\beta+1/s)^2.
\]
Summing over vertices and using $\sum_vq_v=bm$ gives
\begin{equation}
 \Phi_s(q)\le(1-\beta+1/s)^2bm\le\delta^2bm.
 \label{eq:reserve-stable-potential}
\end{equation}
The last inequality uses $1-\beta=\delta/2$.
We also need monotonicity in the witness scale. For a fixed load $q$,
an active summand satisfies
\begin{equation}
 \frac{\partial}{\partial s}\frac{(q+1-\beta s)^2}{s}
 =-\frac{(q+1)^2}{s^2}+\beta^2<0.
 \label{eq:potential-scale-derivative}
\end{equation}
An inactive summand is zero, and the two pieces meet continuously with
zero derivative. Thus increasing $s$ with the allocation fixed cannot
increase potential, including when a coordinate becomes inactive.

\subsubsection{Single-Update Charges}\label{sec:single-update-charges}
\begin{lemma}[Insertion charge]
\label[lemma]{lem:insertion-potential}
A complete insertion transports at most $b$ units and increases the stable
potential by at most $2\delta b$.
\end{lemma}
This bounds both effects of an insertion: the repair performed now and
the potential it can leave for later repairs. Since an inserted edge
adds exactly $b$ units, the first bound is independent of the number of
other edges in the graph.
\begin{proof}[Proof of \Cref{lem:insertion-potential}]
Let $s_i$ and $s_f$ be the initial and final scales. An insertion cannot
decrease witness density, so $s_f\ge s_i$ when the initial state is
nonempty. Fix $s_f$ throughout this argument and define
\[
 h(q)=\frac{(q-\beta s_f+1)_+^2}{s_f}.
\]
Extend it above $s_f$ by its tangent:
\[
 g(q)=\begin{cases}
 h(q),&q\le s_f,\\
 h(s_f)+h'(s_f)(q-s_f),&q>s_f.
 \end{cases}
\]
The tangent extension $g$ is convex and nondecreasing. Its derivative
increases up to $h'(s_f)$ and is constant thereafter, so it is globally
Lipschitz with constant
\[
 \kappa=2(1-\beta+1/s_f)\le2\delta.
\]
The new edge adds nonnegative amounts summing to $b$. Their total
increase in $\sum_vg(q_v)$ is therefore at most $\kappa b$.
Every subsequent transfer obeys
$q_z-a\ge c\ge q_w+a$. For a convex function, transferring $a$ from
the larger coordinate towards the smaller under this condition gives
$g(q_z-a)+g(q_w+a)\le g(q_z)+g(q_w)$.
Thus transfers can only decrease the sum, and witness
changes leave it unchanged because the allocation is fixed.

Initial and final loads are at most $s_f$, where $g=h$.
Writing these vectors as $q^i,q^f$, respectively, yields
\[
 \Phi_{s_f}(q^f)\le \Phi_{s_f}(q^i)+2\delta b
                \le \Phi_{s_i}(q^i)+2\delta b.
\]
The second inequality follows from \Cref{eq:potential-scale-derivative}.
For an initially empty graph, $q^i=0$ and $g(0)=0$, giving the same
bound from initial potential zero. This argument depends on the total
inserted mass $b$, so it also covers mixed edge sizes.

For transport, the old loads satisfy the old capacity and insertion
never lowers that capacity. Adding $b$ units can create at most $b$
units of excess. Each successful path reduces excess by its amount
(\Cref{lem:transfer-progress}); later capacity increases cannot create
excess. Total insertion transport is at most $b$.
\end{proof}

\begin{lemma}[Deletion charge]
\label[lemma]{lem:deletion-potential}
A deletion $d$ with reserve destinations satisfies
\begin{equation}
 \delta A_d\le bZ_d+\Phi_{\rm before}-\Phi_{\rm after},
 \label{eq:deletion-potential-charge}
\end{equation}
where $A_d$ is the transport during that deletion.
\end{lemma}
The repair is paid for by two sources: the new exposure charge $bZ_d$
and a decrease in stored potential. The bound accounts for both in
the same inequality, allowing the potential changes to cancel when
we sum over consecutive updates.
\begin{proof}[Proof of \Cref{lem:deletion-potential}]
Removing the edge decreases every affected load, so it cannot increase
potential at the old scale. If the graph becomes empty, there is no
transport and the final potential is zero, proving the inequality.
Otherwise the raw post-removal vector $\widetilde q$ satisfies
$\widetilde q_v\le s_{\rm old}$ for every vertex.

First consider $s_0<s_{\rm old}$. Both endpoints of this scale interval
are normalized, so $\beta u\ge1$ for $s_0\le u\le s_{\rm old}$.
On an active coordinate, \Cref{eq:potential-scale-derivative} gives
\[
 -\frac{\partial}{\partial u}
   \frac{(q+1-\beta u)^2}{u}
 =\frac{(q+1)^2}{u^2}-\beta^2
 \le\frac{q(q+2)}{u^2}.
\]
The last step uses $1/u^2\le\beta^2$. A coordinate can activate on
this interval only if it belongs to
$J_d=\{v:\widetilde q_v+1>\beta s_0\}$.
Integrating also accounts for coordinates that are initially inactive:
\begin{align*}
 \Phi_{s_0}(\widetilde q)-\Phi_{s_{\rm old}}(\widetilde q)
 &\le\left(\frac1{s_0}-\frac1{s_{\rm old}}\right)
       \sum_{v\in J_d}\widetilde q_v(\widetilde q_v+2)\\
 &\le\left(1+\frac2{s_{\rm old}}\right)
       \left(\frac{s_{\rm old}}{s_0}-1\right)
       \sum_{v\in J_d}\widetilde q_v\\
 &=bZ_d.
\end{align*}
The second step uses $\widetilde q_v+2\le s_{\rm old}+2$.
If $s_0\ge s_{\rm old}$ instead, scale monotonicity gives an increase
of at most zero, agreeing with $Z_d=0$.

After normalization, every witness replacement or peeling refresh
can only increase the scale. Each successful reserve transfer reduces
potential by at least $\delta$ times its amount
(\Cref{lem:reserve-potential-drop}). Adding these changes to the removal
and scale-change bounds proves \Cref{eq:deletion-potential-charge}.
\end{proof}

\subsubsection{Summing the Update Charges}
\label{sec:transport-telescoping}

Insertions add only bounded potential, while the exposure charge
$Z_d$ bounds the increase caused by lowering the scale after a
deletion. These charges pay for all transport in the sequence.

\begin{proof}[Proof of \Cref{lem:hyper-all-deletions}]
An insertion transports at most $b$ and increases stable potential
by at most $2\delta b$. A deletion transporting $A_d$ units satisfies
\[
 \delta A_d\le bZ_d+\Phi_{\rm before}-\Phi_{\rm after}.
\]
These single-update inequalities are proved in
\Cref{lem:insertion-potential,lem:deletion-potential}.
Let $A_+$ and $A_-$ be transport over insertions and deletions,
respectively, and $\Phi_f$ the final stable potential.
Starting from empty gives initial potential zero. Summing the deletion
inequalities cancels intermediate potentials except for the increases
at insertions, yielding
\[
 \delta A_-\le b\sum_d Z_d-\Phi_f+2\delta bN_+.
\]
Since $\Phi_f\ge0$ and $A_+\le bN_+$, we obtain
\[
 A\le3bN_++\frac b\delta\sum_d Z_d.
\]
Now use $b=O(r/\epsilon)$ and
$1/\delta=(1+\epsilon)/\epsilon=O(1/\epsilon)$.
\end{proof}

\subsection{Boundary States and Witness Destruction}
\label{sec:appendix-transport-telescoping}
We now sum the insertion and deletion charges. The intermediate
potentials cancel, leaving only the initial and final stable states.

\begin{lemma}[Transport with boundary potentials]
\label[lemma]{lem:transport-boundary}
Starting from an empty or normalized certified state with
$b\ge\lceil2r/\epsilon\rceil$ and $0<\epsilon\le1$, let $N_+$ be the
number of insertions and $\Phi_0,\Phi_f$ the initial and final stable
potentials. Set the potential of an empty state to zero.
Total transport satisfies
\[
 A\le3bN_++\frac{\Phi_0-\Phi_f+b\sum_d Z_d}{\delta}.
\]
In an $r$-uniform hypergraph, deleting the sole induced edge of a
normalized witness contributes zero exposure.
\end{lemma}

\begin{proof}[Proof of \Cref{lem:transport-boundary}]
Let $A_+$ and $A_-$ be transport summed over insertions and deletions.
By \Cref{lem:insertion-potential}, $A_+\le bN_+$, and the sum of
potential increases over insertions is at most $2\delta bN_+$.
The sum of all update-wise potential changes is $\Phi_f-\Phi_0$.
Summing \Cref{eq:deletion-potential-charge} therefore gives
\begin{align*}
 \delta A_-
 &\le b\sum_dZ_d-\sum_d(\Phi_{\rm after}-\Phi_{\rm before})\\
 &\le b\sum_dZ_d+\Phi_0-\Phi_f+2\delta bN_+.
\end{align*}
Adding $A_+\le bN_+$ and dividing the deletion bound by $\delta$
proves the stated $3bN_+$ bound. It retains both boundary potentials
and applies to every finite prefix of the update sequence.

In an $r$-uniform hypergraph, a one-edge witness has size $k\ge r$,
while normalization gives $1/k\ge1/r$. Hence $k=r$ and its old density
is exactly $1/r$. Any surviving edge supplies at least this density,
so replacing a destroyed witness cannot lower the scale. An empty
result has zero exposure by definition.
\end{proof}

Exposure is confined to a degree tail because
$\widetilde q_v\le b\deg_{\rm after}(v)$. For destruction in mixed
ranks, if the old witness has size $k$ and the fallback edge has size
$f$, then $s_{\rm old}/s_0\le f/k$; in particular, $f\le k$ gives
zero exposure. Thus the charge depends on the exposed degree tail and the drop in
witness density. Combining its transport bound with
\Cref{eq:hyper-aggregate-search} bounds the associated search work.
The new charge need not be smaller on every destructive deletion:
one may replace $Z_d$ there by
$\min\{Z_d,\delta(1+\delta)m_{\rm after}\}$, using the old excess
and stable-potential bounds before telescoping.
Indeed, after removal there are $bm_{\rm after}$ units in total and
capacity only increases during repair, so $A_d\le bm_{\rm after}$.
The final state satisfies $\Phi_{\rm after}\le\delta^2bm_{\rm after}$,
while $\Phi_{\rm before}\ge0$. Consequently
\[
 \delta A_d+\Phi_{\rm after}-\Phi_{\rm before}
 \le\delta(1+\delta)bm_{\rm after}.
\]
Both this inequality and the $bZ_d$ bound concern the same deletion
and the same endpoint potentials, so their right-hand sides may be
minimized before summing. An empty result has $m_{\rm after}=0$.

\subsection{A Degree-Tail Specialization}
\label{sec:degree-tail-specialization}
The exposure bound gives a more explicit guarantee when few vertices
have degree near the repair threshold. Consider an internal deletion
whose old witness has $M\ge2$ edges and $k$ vertices. Set
$M'=M-1$, $s=s_{\rm old}=\alpha bM/k$, and
$s'=\alpha bM'/k$. Let
\[
 H_d=\{v:b\deg_{\rm after}(v)>\beta s'-1\}.
\]
Normalization ensures $s_0\ge s'$, so every coordinate contributing to
$Z_d$ lies in $H_d$. Moreover,
$s/s_0-1\le1/M'$ and $\widetilde q_v\le s$. Hence
\begin{equation}
 Z_d\le\frac{s+2}{bM'}|H_d|
 \le2(1+2\epsilon)\frac{|H_d|}{k}.
 \label{eq:surviving-degree-tail}
\end{equation}
The last step uses $M/M'\le2$ and $1/s\le\delta/2$.

\begin{corollary}[Uniform hypergraphs with a small degree tail]
\label[corollary]{cor:uniform-degree-tail}
Consider an $r$-uniform update sequence satisfying the assumptions of
\Cref{lem:transport-boundary}. If $|H_d|\le\eta|C_d|$ for every
internal deletion that leaves an edge in its old witness, then
\[
 A=O\!\left(bN_++\frac{\Phi_0}{\epsilon}
                +\frac{b\eta}{\epsilon}N_-\right),
\]
where $N_-$ is the number of deletions and $C_d$ is the pre-deletion
witness. In particular, the result applies to ordinary graphs ($r=2$).
\end{corollary}
For an empty start and $b=\Theta(r/\epsilon)$, the bound is
$O((r/\epsilon)N_+ +(r\eta/\epsilon^2)N_-)$ units of transport.
The parameter $\eta$ measures the size of the relevant degree tail
relative to the witness, so this conclusion is useful when that ratio
stays small throughout the sequence.
\begin{proof}
External deletions leave the witness density unchanged. In a uniform
hypergraph, destruction of a normalized one-edge witness has zero
exposure by \Cref{lem:transport-boundary}. For every remaining deletion,
\Cref{eq:surviving-degree-tail} gives $Z_d=O(\eta)$.
Substitute their sum in \Cref{lem:transport-boundary}, drop $\Phi_f\ge0$,
and use $1/\delta=(1+\epsilon)/\epsilon=O(1/\epsilon)$.
\end{proof}
For mixed ranks, the same bound retains the additional term
$(b/\delta)\sum_{d:M_d=1}Z_d$ for witness destruction.
The degree-tail condition permits many low-degree vertices; it concerns
the degree threshold at each deletion and need not follow from a fitted
power-law exponent. This is a transport bound; search and endpoint costs
are incorporated in \Cref{sec:total-maintenance-work}.

\section{Search Counts and Complete Time Accounting}
\label{sec:maintenance-details}
We first bound the number of searches over an update sequence, then add
the costs of normalization, heap maintenance, and occasional peeling.
Throughout this subsection, $h$ bounds the high-load region at every
search and $\mathcal I=I$ denotes the active incidence volume.

\paragraph{Making the search count explicit.}
Start from an empty state or a normalized certified state, and let
$\bar\rho\ge1/r$ bound the optimum over all states, including the initial
state. Before an update, stability gives $q_v\le\alpha b\bar\rho$.
An insertion adds only $b$ units, while a deletion decreases loads.
During repair, a transfer decreases the source load, leaves intermediate
loads unchanged, and keeps its destination at most $c\le\alpha b\bar\rho$.
Thus every search has $Q\le\alpha b\bar\rho+b$.
For an internal deletion let $M_d,k_d$
be the old witness's edge count and size. Define
\[
V_a=\sum_{M_d\ge2}\left\lceil\frac{ab}{k_d}\right\rceil
       +\sum_{M_d=1}\left(\left\lfloor\frac{ab}{k_d}\right\rfloor
                          -\left\lfloor\frac{ab}{r}\right\rfloor\right).
\]
The quantity $V_a$ bounds downward changes of the normalized integer
target, assigning an empty state the artificial baseline $\lfloor ab/r\rfloor$.
For nonempty states this target is $\lfloor abL\rfloor$. The first sum
covers a witness that loses one of several edges; the second covers loss
of its sole edge and subsequent normalization or an empty state.
For transport we need only decreases before nonempty repairs, since an
emptying update requires no transport. We use $V_\alpha$ for those outer
capacity decreases and $V_\gamma$ for the destination target.
\begin{lemma}[Transport and failed-search counts]
\label[lemma]{lem:search-counts}
For \method{}, if the high-load region has at most $h$
vertices at every search, then
\[
 A\le bN_++hV_\alpha,
\]
\[
F\le\lfloor\gamma b\bar\rho\rfloor-\lfloor\gamma b/r\rfloor+V_\gamma.
\]
\end{lemma}
The first bound charges successful transport to inserted mass and
capacity decreases. The second charges failed searches to upward
threshold progress: each failure increases the threshold, and only
deletions can make that progress need to be repeated. Together they
supply search counts for \Cref{eq:hyper-aggregate-search}.
\begin{proof}[Proof of \Cref{lem:search-counts}]
By \Cref{lem:insertion-potential}, insertion transport is at most $b$.
For a deletion requiring repair, let $c_{\rm old}$ be its old capacity
and $c_{\rm start}$ the capacity at its first search, after normalization
and any initial peeling. Raw removal only decreases loads, so every
load is at most $c_{\rm old}$. At this first search, every overloaded
vertex belongs to its high-load region, which contains at most $h$
vertices. The initial excess is therefore at most
$h(c_{\rm old}-c_{\rm start})_+$. During repair, capacities only increase
and each successful path reduces excess by its amount. This initial
excess bounds all transport during the deletion.

If $M_d\ge2$, the old witness retains density $(M_d-1)/k_d$ before
normalization. Subsequent normalization can only improve it. For any
$a>0$, the threshold decrease is consequently at most
\begin{align*}
 \left\lfloor\frac{abM_d}{k_d}\right\rfloor
 -\left\lfloor\frac{ab(M_d-1)}{k_d}\right\rfloor
 \le\left\lceil\frac{ab}{k_d}\right\rceil.
\end{align*}
This follows from $\lfloor x\rfloor-\lfloor x-y\rfloor\le\lceil y\rceil$.
If $M_d=1$ and an edge survives, normalization gives density at least
$1/r$, so the decrease is at most
$\lfloor ab/k_d\rfloor-\lfloor ab/r\rfloor$.
This quantity is nonnegative since the old normalized witness has
$k_d\le r$. External deletions do not decrease the witness density,
and deletion of the last active edge requires no transport. Taking
$a=\alpha$ and summing proves the bound on $A$.

For failed searches, define $t_*=\lfloor\gamma b/r\rfloor$ and track
$\tau=t$ in nonempty states and $\tau=t_*$ in empty states.
Normalization always gives $t\ge t_*$. Every failure increases this
integer target by at least one (\Cref{eq:failure-threshold-growth}).
Insertions and density-improving witness replacements cannot decrease it.
The same floor-difference calculation with $a=\gamma$ bounds its total
downward variation by $V_\gamma$. In particular, on an emptying update
the previous state has one active edge, and its target drops by at most
$\lfloor\gamma b/k_d\rfloor-t_*$; the next insertion cannot take it
below $t_*$. Thus empty states introduce no additional charge.
Since $t_*\le\tau\le\lfloor\gamma b\bar\rho\rfloor$, summing target
increases and decreases gives the claimed bound on $F$.
\end{proof}

\Cref{lem:transport-boundary} gives an alternative bound on $A$ for
reserve repair; the smaller bound can be substituted into
\Cref{eq:hyper-aggregate-search}. A fixed-size high-degree core with all other vertices below the
actual repair threshold is a sufficient structural regime. Crossing edges
are included in the degrees. Snapshot density separation alone does not
establish this condition.

\paragraph{Fallback and maintenance costs.}
Failed searches and peeling supply known induced-edge counts.
For an unknown-count fallback $S$, a separate edge stamp deduplicates
encounters in its positive incidence lists before testing containment.
The test stops at the first endpoint outside $S$. With analogous
counters $P_S,E_S,J_S$,
\[
 P_S\le\sum_{v\in S}q_v,\qquad E_S\le P_S,\qquad
 J_S\le\min\{r,|S|+1\}E_S.
\]
Every internal edge has a positive owner in $S$ and is counted once.
The fallback edge itself is fully tested, so $J_S\ge|S|$ also pays
for copying and marking the new witness.
We retain
$Z_{\rm norm}=\sum_S(P_S+J_S)$, as mixed-rank fallback need not cost
$O(b)$.

Let $R_{\rm up}=\sum_{\rm updates}|e|$, $K$ count successful searches,
$\Lambda=\log(2+n+\mathcal I_{\max})$, $D_0$ be inherited refresh debt,
and $\Pi_0$ be initial container-growth credit. For each successful path,
$h_{\rm path}$ is the size of $H$ at its search. With compact or
preallocated edge IDs, the single-path implementation has the following
detailed accounting. Its first line pays for inherited state credits and edge
updates, its second for search inspections, fallbacks, and heap changes,
its third for locating receiving shares, and its last for peeling:
\begin{equation}
\begin{aligned}
O\bigl(&\Pi_0+|C_0|+
 R_{\rm up}[\log(r+1)+\log(b+1)+\log(n+1)]\\
 &+P_{\rm tot}+J_{\rm tot}+Z_{\rm norm}+K\log(n+1)\\
 &+\sum_{\rm paths}[(\ell-1)\log(h_{\rm path}+1)+\log(r+1)]\\
 &+(P_{\rm tot}+D_0)\Lambda\bigr).
\end{aligned}
\label{eq:total-maintenance-work}
\end{equation}
We justify each term in \Cref{eq:total-maintenance-work} below.

\paragraph{Edge updates and path transfers.}
For an updated edge of size $s_e$, sorting and locating endpoint slots
cost $O(s_e\log(r+1))$. Discrete water filling searches at most $b+1$
integer levels, scanning the $s_e$ endpoints in each binary-search step,
and costs $O(s_e\log(b+1))$. Heap changes contribute
$O(s_e\log(n+1))$. Summing $s_e$ gives the $R_{\rm up}$ term.
Positive-incidence lists use inverse positions, so changing one endpoint
share costs amortized constant time apart from its slot lookup.
The initial container-growth credit $\Pi_0$ pays any inherited allocation
cost not charged to updates in the sequence.

Search exploration contributes $P_{\rm tot}+J_{\rm tot}$. Its parent
pointers already identify donor slots. A receiving slot is found by
binary search in the sorted endpoints.
As in the proof of \Cref{lem:local-search-work}, every
nonterminal path edge was fully expanded before a destination was found,
so all of its endpoints lie in the high-load region. Only the terminal
edge can have rank exceeding $h_{\rm path}$, giving the displayed
sum of lookup costs. Intermediate load changes
cancel, so a successful path changes only two heap keys, contributing
$O(\log(n+1))$ per success. Traversing the parent path itself is paid
by the search that discovers it.

\paragraph{Witness installation and budgeted peeling.}
A failed search already supplies its reached vertices and induced-edge
count. A fallback costs $P_S+J_S$, and a peeling pass supplies its new
witness as part of that pass. The total cost of clearing old membership
marks is at most the initial $|C_0|$ plus the total sizes of subsequently
installed witnesses. These sizes are charged to the corresponding
searches, fallbacks, or peeling passes.

At a peeling refresh $j$, let $D_j$ be the accumulated positive-scan
debt, and let $m_j,\mathcal I_j$ describe the current graph.
The trigger requires $D_j>4(n+rm_j)$. Since
$\mathcal I_j\le rm_j$, a pass costs
\[
 O((n+\mathcal I_j)\log(2+n+\mathcal I_j))=O(D_j\Lambda).
\]
Debt is reset after each pass, so the charged scan intervals are disjoint
and $\sum_jD_j\le P_{\rm tot}+D_0$. This proves the final term of
\Cref{eq:total-maintenance-work}, even if the graph size changes between
refreshes. It also includes constructing and marking the new peeling
witness.

\paragraph{Construction, reporting, and space.}
An empty state takes $\Theta(n)$ time to construct. For a nonempty
initial graph, we also charge the work of inserting its initial edges.
Query materialization adds $\Theta(1+|C|)$ per reported set.
The bounds use representable unit-cost integer arithmetic and compact
or preallocated edge identifiers, with allocation charged to the update
sequence. A compacted state occupies $O(n+I)$ words. The
prototype retains deleted edge vectors and allocated container capacities,
so its storage also depends on the update history. Independent correctness
checking is performed outside the timed updates.

\paragraph{Deriving the summary guarantees.}
\begin{proof}[Proof of \Cref{thm:amortized-update,thm:main-guarantee}]
Starting empty gives $C_0=\emptyset$, $D_0=0$, and no inherited
container-growth credit; initialization costs $O(n)$. Write
$S_{\rm scan}=P_{\rm tot}+J_{\rm tot}$. Every successful search inspects
an incidence, so $K\le S_{\rm scan}$. The incidences along a returned
path have already been examined by that search, so the total number of
path steps is also at most $S_{\rm scan}$.
The nonempty update sequence gives $m_*\ge1$.
Since $I_*\le r m_*$ and $b=\lceil2r/\epsilon\rceil$,
all logarithms in \Cref{eq:total-maintenance-work} are
$O(\log((n+r m_*)/\epsilon))$. Substitution yields
\Cref{eq:readable-total-work}, including the cost of peeling.

For the coarse bound, a search inspects at most $I_*$ positive incidences
and at most $I_*$ endpoints: each vertex is expanded once and each edge
is stamped after its first expansion. Its inspection cost is at most
$2I_*$. An insertion transports at most $b$ units. During a deletion,
the initial excess is at most the remaining total mass $bm_*$;
capacities only increase during repair, so deletion transport is at most
$bm_*$. Integer augmentation therefore bounds successful searches per
update by $bm_*$.

Each failed search increases the inner threshold by a factor of at least
$1+\epsilon/4$ (\Cref{eq:failure-threshold-growth}). The initial threshold
is at least $\lfloor\gamma b/r\rfloor$ and the final one is at most
$\lfloor\gamma bm_*\rfloor$. Hence there are
$O(\epsilon^{-1}\log(2+rm_*))$ failures per update. Consequently,
\[
 S_{\rm scan}=O\!\left(NI_*\left[bm_*
                +\epsilon^{-1}\log(2+rm_*)\right]\right).
\]
There is at most one fallback normalization per update. Its stamped
incidence and endpoint scans each cost at most $I_*$, giving
$Z_{\rm norm}=O(NI_*)$, while $R_{\rm up}\le Nr$.
Finally, $b=O(r/\epsilon)$ and $\log(2+rm_*)=O(rm_*)$.
Substitution in \Cref{eq:readable-total-work} gives
\begin{equation}
 T_{\rm maint}=O\!\left(n+N\frac{r m_* I_*}{\epsilon}
 \log\frac{n+r m_*}{\epsilon}\right).
 \label{eq:general-update-bound}
\end{equation}
Excluding initialization and dividing by $N$ proves
\Cref{thm:amortized-update}. The approximation follows from
\Cref{thm:certificate}; the density and witness are stored explicitly,
giving the stated query costs.
\end{proof}

The main text proves \Cref{lem:combined-update-time} by substituting
the locality and transport bounds into \Cref{eq:readable-total-work}.

%% file: sections/cq_reconstruction.tex
\begin{revision}
\section{CQ Reconstruction and Evaluation Protocol}
\label{app:cq-reconstruction}
Our \texttt{CQ} baseline reconstructs the basic amortized graph core,
nested label buckets, and shared-copy representation from a preliminary
version of Chekuri et al.~\cite{chekuri2024}, using the hypergraph
extension described in its Section~6. This preliminary version is linked
in \Cref{sec:experiment-setting}. Each logical insertion or deletion
executes $K$ unit-copy operations; copies with the same hyperedge and
head share a multiplicity and endpoint labels. There are no bulk
transfers or empirically tuned copy counts.

\paragraph{Fixed parameters and extraction.}
Let $n$ be the full stream's vertex-universe size and $r$ its maximum
hyperedge rank. For requested $0<\epsilon\le1$, we set
\begin{align*}
q&=(1+\epsilon)^{1/4},
& M&=\left\lceil\frac{\log\max\{2,n\}}{\log q}\right\rceil,\\
a&=\exp\!\left(\frac{\log(1+\epsilon)}{8M}\right)-1,
& P&=(1+a)^2,\\
\beta&=3+3a,
& \Delta&=(1+\epsilon)^{-1/2}-(1+\epsilon)^{-1},
\end{align*}
and $K=\lceil\beta Mr/(q\Delta)\rceil+1$.
These explicit constants are our conservative choice, not constants
specified in the source paper.
At quiescence, the label invariants imply
$d(h)\le P\,d(u)+3+2a\le P\,d(u)+\beta$
for each copy head $h$ and every endpoint $u$.
With $\mu=\max_v d(v)$, start at $t_0=\mu$ and use
$t_{i+1}=(t_i-\beta)/P$ and $S_i=\{v:d(v)\ge t_i\}$,
including all threshold ties.
We return $S_{i+1}$ at the first index with
$|S_{i+1}|\le q|S_i|$ and compute its fully induced logical-edge density.
For a nonempty snapshot satisfying these invariants,
\[
\rho(S_{i+1})\ge
\frac{\rho^*}{qP^M}-\frac{\beta M}{qK}
\ge\frac{\rho^*}{1+\epsilon},
\]
using $\rho^*\ge1/r$; empty snapshots return the empty set.

\paragraph{Deletion completion and implementation scope.}
Section~6 of that preliminary version does not specify hypergraph deletion pseudocode.
Our completion selects a global minimum-load endpoint for a unit
reorientation, relabels the destination group, and schedules both changed
endpoints. A triggering endpoint whose condition remains violated stays
on the worklist; residual source-group copies retain their labels.
We have not established general termination or the original amortized
update-time bound for this worklist. We do not implement the deamortized
or full worst-case construction.
Direct group-specific endpoint labels and links require
$O(\sum_{e\in E}|e|^2)$ worst-case live storage, in addition to the vertex
state and registry capacity determined by peak simultaneous groups.
Retired group identifiers are reused.
Extraction scans at most $M$ thresholds and the active edge incidences,
taking $O(Mn+p)$ time for $p=\sum_{e\in E}|e|$.
Thus, we do not claim the original algorithm's time or space bounds.

\paragraph{Result admission.}
Each reported \texttt{CQ} mean requires three complete repetitions,
with the entire update stream and every scheduled query processed.
Independent replay checks the returned sets' induced densities against
exact snapshot optima and the requested rational accuracy.
Passing prefixes and single
decision trials do not constitute three-run results.
The 1,800-second \texttt{CQ} wall limit includes input loading,
invariant checks, and allocation output; these are excluded from running
time. We sample working-set and private-memory use against a 24\,GiB
limit rather than impose a hard operating-system cap.
Running-time means follow the reporting criteria in
\Cref{sec:experiment-setting}.
\end{revision}